\documentclass[12pt]{elsarticle}
\usepackage[utf8]{inputenc}
\usepackage[T1]{fontenc}
\usepackage[english]{babel}
\usepackage{lmodern}
\usepackage{colortbl}
\usepackage{lastpage}
\usepackage{amscd}
\usepackage{algpseudocode}
\usepackage{amsmath}
\usepackage{multirow}
\usepackage{amsfonts}
\usepackage{amssymb}
\usepackage{latexsym}
\usepackage{mathrsfs}
\usepackage{float}
\usepackage{wrapfig}
\usepackage{color,framed}
\usepackage{longtable}
\usepackage{graphics}
\usepackage{graphicx}
\usepackage{acronym}
\usepackage{comment}
\usepackage{listings}
\usepackage{mathdots} % for dotted math symbols
\usepackage[left=3cm,right=3cm]{geometry}
\usepackage{ulem} % for underlining text
\usepackage[colorlinks=true,linkcolor=blue,urlcolor=violet,bookmarks=true]{hyperref} % for hyperlinks
\usepackage{array}

\newcolumntype{L}[1]{>{\raggedright\let\newline\\\arraybackslash\hspace{0pt}}m{#1}}
\newcolumntype{C}[1]{>{\centering\let\newline\\\arraybackslash\hspace{0pt}}m{#1}}
\newcolumntype{R}[1]{>{\raggedleft\let\newline\\\arraybackslash\hspace{0pt}}m{#1}}

\usepackage{color,xcolor,soul}

\def\dsum #1#2{\displaystyle{\sum_{#1}^{#2}}}
\def\dprod #1#2{\displaystyle{\prod_{#1}^{#2}}}

\newcommand{\Fq}{\mathbb{F}_{q}}

\DeclareMathOperator{\Tr}{Tr}

\DeclareMathOperator{\lcm}{lcm}

\usepackage{ntheorem}
\theoremstyle{break}
{\theorembodyfont{\upshape}
    \newtheorem{theorem}{Theorem}[section] % Main counter
    \newtheorem{lemma}[theorem]{Lemma} % Shares numbering with theorem
    \newtheorem{proposition}[theorem]{Proposition}
    \newtheorem{corollary}[theorem]{Corollary}
    \newtheorem{definition}[theorem]{Definition}
    \newtheorem{example}[theorem]{Example}
    \newtheorem{remark}[theorem]{Remark}
    
    \newtheorem*{proof}{{Proof}}

}
\title{\LARGE \bfseries
Equivalence of Families of Polycyclic   Codes over Finite Fields}

\begin{document}
	
 \begin{frontmatter}

\author[nm]{Hassan Ou-azzou\corref{cor1}} %\corref{cor1}
    \ead{hassan.ouazzou@student.unisg.ch} 
 \author[nm]{Anna-Lena Horlemann}
     \ead{anna-lena.horlemann@unisg.ch}

                   \cortext[cor1]{Principal corresponding author}
        		\address[nm]{ School of Computer Science, University of St.Gallen, St.\ Gallen, Switzerland }

\begin{abstract}
We study the equivalence of families of polycyclic codes associated with polynomials of the form  $x^n - a_{n-1}x^{n-1} - \ldots - a_1x - a_0$ 
over a finite field. We begin with the specific case of polycyclic codes associated with a trinomial $x^n - a_{\ell} x^{\ell} - a_0$ (for some $0< \ell <n$), which we refer to as \textit{$\ell$-trinomial codes},  after which we generalize our results to general polycyclic codes. We introduce an equivalence relation called \textit{$n$-equivalence}, which extends the known notion of $n$-equivalence for constacyclic codes \cite{Chen2014}. We compute the number of $n$-equivalence classes for this relation and provide conditions under which two families of polycyclic (or $\ell$-trinomial) codes are equivalent. In particular, we prove that when $\gcd(n, n-\ell) = 1$, any $\ell$-trinomial code family is equivalent to a trinomial code family associated with the polynomial $x^n - x^{\ell} - 1$. Finally, we focus on $p^{\ell}$-trinomial codes of length $p^{\ell+r}$, where $p$ is the characteristic of $\Fq$ and $r$ an integer, and provide some examples as an application of the theory developed in this paper.

\end{abstract}
        	\end{frontmatter}

\section{Introduction}

Algebraic coding theory is concerned with the design of codes that ensure reliable communication over noisy channels. A foundational contribution to this field is Claude Shannon’s seminal paper \textit{A Mathematical Theory of Communication} \cite{Shannon1948}, which established the theoretical limits of reliable transmission. Although Shannon proved the existence of codes with arbitrarily good performance, his work did not provide explicit constructions. Consequently, the problem of constructing codes with optimal or near-optimal parameters has become a central challenge in coding theory.

The pioneering works of Shannon, Hamming, and Golay in the mid-twentieth century \cite{Shannon1948,hamming1950error,golay1949notes} laid the foundations of the field. Since then, coding theory has developed strong connections with several areas of pure mathematics, most notably algebra and combinatorics. In particular, algebraic techniques have played a crucial role in the construction and analysis of efficient codes.\\
A linear code of length $n$ over a finite field $\Fq$ is defined as an $\Fq$-vector subspace of $\Fq^n$. Such a code is characterized by three parameters: its length $n$, dimension $k$, and minimum distance $d$. One of the fundamental problems in the theory of linear codes is to maximize the minimum distance $d$ for fixed values of $n$ and $k$, and to construct codes achieving this optimum by algebraic means. 
While extensive computer searches have been used to identify linear codes with the best known parameters, this approach is inherently limited. In particular, determining the minimum distance of a linear code is known to be computationally intractable \cite{vardy2002intractability}, making exhaustive searches infeasible for large parameters. As a result, much research has focused on special families of linear codes endowed with rich algebraic structure. 
Among these, cyclic codes—introduced by Prange in the 1950s \cite{Prange1957}—are one of the most extensively studied classes, owing to both their theoretical significance and practical applicability. Cyclic codes admit a natural algebraic description as ideals in the quotient ring $\Fq[x]/\langle x^n - 1 \rangle$, and many codes with good parameters are either cyclic or closely related to them. 

Cyclic codes were later generalized to constacyclic codes \cite{Berlekamp1968} and further to polycyclic codes, also known as pseudo-cyclic codes \cite{10}. 
Like cyclic and constacyclic codes, polycyclic codes over $\Fq$ can be represented as ideals in the quotient ring $\Fq[x]/\langle f(x) \rangle$, where $f(x)$ is a nonzero polynomial in $\Fq[x]$. Polycyclic codes specialize to constacyclic codes when $f(x) = x^n - \lambda$ for some nonzero $\lambda \in \Fq$, and further to cyclic codes when $\lambda = 1$ and negacyclic codes when $\lambda = -1$. Polycyclic codes have attracted increasing attention in the literature; see, for example, \cite{Aydin2022,12,Ouazzou2021,Shi2020,Shi2023}.

Another tool commonly used in coding theory is the study of code equivalences, where two linear codes $C_1$ and $C_2$ of length $n$ over $\Fq$ are said to be equivalent if there exists an isometry $f$ of $\Fq^n$ such that $f(C_1)=C_2$. Knowing equivalences is particularly useful for reducing the complexity of computer searches for good codes \cite{aydin2023new}, as well as for applications in code-based cryptography, 
because 
%understanding equivalence between ambient spaces is particularly important, as different defining polynomials may give rise to equivalent families of codes. 
identifying such equivalences helps avoid redundant or insecure parameter choices and provides insight into potential structural vulnerabilities, thereby contributing to both the design and the cryptanalysis of code-based cryptosystems. %\cite{mceliece1978public}.  

In the context of code equivalence, Chen et al.\ introduced in \cite{Chen2014} a systematic approach to the equivalence of ambient spaces of constacyclic codes by defining an equivalence relation, called \textit{$n$-equivalence} and denoted by $\sim_n$, on the set $\Fq^\ast$ of nonzero elements of $\Fq$. For $\lambda,\mu \in \Fq^\ast$, one has $\lambda \sim_n \mu$ if there exists $a \in \Fq^\ast$ such that the map
$$
\Phi_a : \mathbb{F}_q[x]/\langle x^n-\mu\rangle \longrightarrow \mathbb{F}_q[x]/\langle x^n-\lambda\rangle,
\qquad
\Phi_a(f(x)) = f(ax),
$$
is an $\Fq$-algebra isomorphism that preserves the Hamming distance.\footnote{In particular, every code in $\mathbb{F}_q[x]/\langle x^n-\mu\rangle$ has an equivalent code in $\mathbb{F}_q[x]/\langle x^n-\lambda\rangle$, and conversely.}
Equivalently, $\lambda \sim_n \mu$ if the polynomial $\lambda x^n-\mu$ has a root in $\Fq$. 
If $C=\langle g(x)\rangle$ is a $\lambda$-constacyclic code, then the generator polynomial of the equivalent code $\Phi_a(C)$ is given by $\langle g(ax)\rangle$. More recently, this notion of equivalence has been extended to skew constacyclic codes over $\Fq$ \cite{Ouazzou2025} and to constacyclic codes over finite chain rings \cite{Chibloun2024}. 
For polycyclic codes, Aydin et al.\ \cite{Aydin2022} investigated trinomial codes and formulated several conjectures concerning equivalence and duality. These conjectures were subsequently settled in \cite{Shi2023}, where explicit constructions of isodual and self-dual polycyclic codes were also obtained.

In this paper, we continue the study of polycyclic codes by extending the notion of $n$-equivalence to polycyclic codes associated with polynomials of the form
\[
x^n - a_{n-1}x^{n-1} - \cdots - a_1x - a_0
\]
over the finite field $\mathbb{F}_q$. 
The study of $n$-equivalence for polycyclic codes is motivated by the need to organize and classify these codes in a systematic way. By understanding which ambient spaces are equivalent, one can reduce the search space in future investigations for codes with optimal parameters, and this classification also provides a foundation for further theoretical results on polycyclic codes. 
Within this general framework, we focus first on $\ell$-trinomial codes, which represent the simplest class of polycyclic codes not contained in previously studied families of cyclic or constacyclic codes. Analyzing this case first allows us to illustrate the behavior of $n$-equivalence in a concrete setting and provides insight that guides the extension of our results to more general polycyclic codes.

Our approach is motivated by the fact that trinomial and polycyclic codes generalize cyclic and constacyclic codes, as well as by the successful application of this method to constacyclic codes. Moreover, studying this equivalence for trinomial and polycyclic codes allows us to reduce the search for good codes by restricting attention to codes arising from representatives of each equivalence class. 
We begin with the special case of polycyclic codes associated with the trinomial $x^n - a_{\ell}x^{\ell} - a_0$, which we refer to as \textit{$\ell$-trinomial codes}, and then generalize our results to arbitrary polycyclic codes. We determine the number of $n$-equivalence classes and establish conditions under which two ambient spaces of polycyclic (or $\ell$-trinomial) codes are equivalent. In particular, we show that if $\gcd(n,n-\ell)=1$, then any $\ell$-trinomial code is equivalent to an ambient space associated with the trinomial $x^n - x^{\ell} - 1$.

The remainder of this paper is organized as follows. Section \ref{sec:prelim} provides a review of the basic background on polycyclic codes, and  we prove some necessary results on binomial polynomials that will be used in our study.  In Section \ref{sec:eq_trinomial} we study the properties of the  $(n,\ell)$-equivalence relation and provide conditions under which two $\ell$-trinomial   code families are equivalent. 
In  Section \ref{sec:eq_trinomial_special},  we focus on $p^{\ell}$-trinomial codes of length $p^{\ell+r}$, where $p$ is the characteristic of $\Fq$ and $r$ an integer. Afterwards, we provide some examples as an application of the theory developed in this paper in Section \ref{sec:examples}.
Finally, in Section \ref{sec:eq_poly}, we  generalize our results to general polycyclic codes and provide  conditions to obtain equivalence between two two families of polycyclic codes. 
We conclude this work in Section \ref{sec:conclusion}.

\section{Preliminaries}\label{sec:prelim}
In this section, we recall some basic definitions and properties of polycyclic codes.  Let $\mathbb{F}_q$ be the finite field of order $q$ where $q=p^s$  for a prime $p$ and a positive integer $s$. A \textit{linear code} $C$ of length $n$ over $\Fq$ is an $\Fq$-subspace of $\Fq^{n}$.   We define the \textit{Hamming weight} $w_H(c)$ 
%of $C$ as 
%$$w_H(C):=\min \left\{w_{_H}(c) \mid c \in C, c \neq 0\right\},$$ 
%where $w_{_H}(c)$ is
as the number of nonzero components of $ c=(c_0,c_1,\ldots,c_{n-1})\in \Fq^{n}$. The Hamming distance $d(c,c')$ between two vectors $c$ and $c'$ is defined as $d(c, c')=\left|\left\{i \mid c_i \neq c'_i\right\}\right| = w_H(c-c')$. 
The \textit{(minimum) Hamming distance} of a code $C$ is defined as
$$ d(C):=\min \left\{d(c, c') \mid c \neq c'\right\}.$$ 
It is well known and easy to see that for a linear code $C$ we have $d(C)=w_H(C):=\min \left\{w_{_H}(c) \mid c \in C, c \neq 0\right\}$.  By $[n,k,d]_q$ we denote  a linear code $C$ over $\Fq$ of length $n$, dimension $k$, and minimum distance (at least) $d$. 
	\begin{definition}[Polycyclic codes]  \label{D1}	
		Let   $ C $  be a linear code  of length $n$ over $\Fq$ and  $\vec{a}=\left(a_{_0},a_{_1},\ldots, a_{_{n-1}}\right)\in \Fq^{n}$. We say that $C$ is	
		\begin{enumerate}	
			\item  a \textit{right polycyclic code} with associated vector  $\vec{a} $ if for each  codeword	$\\ c=\left(c_{_0},c_{_1},\ldots, c_{_{n-1}}\right)\in C $ we have $ \left(0,c_{_0},\ldots, c_{_{n-2}}\right)+c_{_{n-1}} \vec{a} \in C$,
			\item  a \textit{left   polycyclic code} with associated vector  $\vec{a} $  if for each  codeword $\\ c=\left(c_{_0},c_{_1},\ldots, c_{_{n-1}}\right)\in C$ we have $ \left(c_{_1},\ldots, c_{_{n-1}},0\right)+c_{_{0}} \vec{a} \in C$,
   \item  a \textit{bi-polycyclic code} with associated vector  $\vec{a} $ if it is both a right and a left polycyclic code with associated vector  $\vec{a}. $ 
		\end{enumerate}			
	\end{definition} 
\begin{remark}\label{R1}
		For any $\lambda \in \Fq^*,$ the $\lambda$-constacyclic codes are the right polycyclic codes associated with the vector $\vec{a}=\left(\lambda,0,\ldots, 0\right)$, and the left polycyclic codes associated with the vector $\vec{b}=\left(0,0,\ldots, \lambda\right). $ In particular,     cyclic codes ($ \lambda=1 $) and  negacyclic codes  ($ \lambda=-1 $) are special cases of polycyclic codes.\end{remark}
 \begin{definition}[Trinomial codes, \cite{Aydin2022}]
    Let $\vec{a}=\left(a_{_0},0,\ldots, 0,a_{\ell}, 0, \ldots, 0\right)\in \Fq^{n},$ with $a_0\neq 0 $ and $ a_{\ell}\neq 0$, for some $ 0<\ell <n. $
    We say that $C$ is a \textit{$\ell$-trinomial code}  of length $n$ over $\Fq$ if it is a (right) polycyclic code with associated vector $\vec{a}.$
\end{definition}

Let $\Fq[x]/ \langle f(x)\rangle,$ with $ \langle f(x)\rangle $ being the  ideal of $\Fq[x]$ generated by $g(x).$ In this work, we mainly work with \textit{right polycyclic} codes, which we henceforth simply refer to as \textit{polycyclic} codes. Under the usual identification of vectors with polynomials, 
   	\begin{equation}\label{Realization}
  	\begin{array}{rccc}
  	\varphi \ : \ & \ \Fq^n &\longrightarrow & \Fq[x]/ \langle f(x)\rangle \\ & & & \\
  	  &v=(v_0, v_{_1},\ldots,v_{_{n-1}}) & \longmapsto &  v(x)=\dsum{i=0}{n-1} v_{_i} x^{i},  
  	\end{array}
  	\end{equation} 
each polycyclic code $C$ of length $n$  associated with a vector $\vec{a}$ is seen as an ideal    in  $ \Fq[x]/ \langle f(x)\rangle,$  where $f(x)=x^n-\vec{a}(x)=x^n-a_{n-1}x^{n-1}-\ldots-a_1x-a_0.$

In the following proposition we collect some basic results on (right) polycyclic codes.
	\begin{proposition}[\cite{Lopez2009}]\label{P2}
		Let $ C\subseteq \Fq^{n} $ be a polycyclic code with associated vector $\vec{a}=\left(a_{_0},a_{_1},\ldots, a_{_{n-1}}\right)$. Then we have the following assertions:
		\begin{enumerate}
		\item The set $ \varphi(C)$ is an ideal of the polynomial ring $\Fq[x]/\langle x^{n} - \vec{a}(x) \rangle, $ with $ \vec{a}(x)=a_{n-1}x^{n-1}+\ldots+a_1x+a_0.$
		\item There is a monic polynomial  $ g(x)\in \Fq[x]  $ of least degree which  divides $  x^{n} - \vec{a}(x) $ 
			and   $ \varphi(C)=\langle g(x)\rangle. $
			\item The set $ \{ g(x),xg(x),\ldots, x^{{n-\deg(g)-1}}g(x)\} $ forms a basis of $ \varphi(C) $ and the dimension of  $ C $ is $ n-\deg(g(x)). $
			\item A generator matrix $ G $ of $ C $ is given by :\\
			$$ G=\small{\left(\begin{array}{c}
			\varphi^{-1}(g(x)) \\ 
			\\
			\varphi^{-1}\left(xg(x)\right)\\ 
			\vdots \\
			\vdots\\
			\varphi^{-1}\left(x^{{k-1}}g(x)\right)
			\end{array} \right) =\left( \begin{array}{cccccccc}
			g_{_0} & g_{_1}  &  \ldots&g_{_{n-k}}& 0        &\ldots  & \ldots & 0 \\ 
			\\
			0      & g_{_0} &g_{_1}  &  \ldots  &g_{_{n-k}}& 0      & \ldots & 0  \\ 
			\vdots & \ddots & \ddots & \ddots   &          & \ddots &        & \vdots \\ 
			\vdots &        & \ddots & \ddots   & \ddots   &        & \ddots & \vdots  \\ 
			0      &\ldots &        &0         & g_{_0}   & g_{_1} & \ldots & g_{_{n-k}}   \\ 
			\end{array}  \right)}$$
		\noindent	where  $k=n-\deg(g(x))$ and $ g(x)=  \dsum{i=0}{n-k} g_{_i}x^{i}.$
				\end{enumerate}
		\end{proposition}

\begin{remark}
    From now to the end of this work we identify $C$ with the ideal  $\varphi(C)$ in $ \Fq[x]/\langle f(x)\rangle.$ 
 \end{remark}   
    We now give the definition of generator polynomials for polycyclic codes. Like a cyclic code, a right polycyclic code has many generators, but among all its generators there is a special unique one, called the standard generator of $C$. Any other generator of $C$ is a multiple of the standard generator. 
    \begin{definition}
    Let $C \subseteq \Fq^{n}$ be a non-zero (right) polycyclic code of length $n$ over $\Fq$. Then the \textit{standard generator} of $C$ is the monic polynomial $g(x)$ of least degree in $ \Fq[x]/\langle f(x) \rangle $ such that $C=\langle g(x) \rangle $.  In this paper, we refer to the standard generator of $C$ as ''the generator''. 
    \end{definition}

Furthermore, we recall the following lemmas concerning the number of solutions of a binomial equation in $\mathbb{F}_q$ and the degree of $\gcd\left(x^n - a, x^m - b\right)$.

\begin{lemma}\cite[Lemma 1]{Schwarz1949}\label{L_soltion}
Consider the finite field $\Fq$ and $ a \in \Fq^*$. Let $n>0$ be an integer and $d=\gcd(n, q-1)$. Then the equation
$$ x^n-a=0$$
has solutions in $\Fq$  if and only if 
$$a^{\frac{q-1}{d}}=1 ,
$$
in which case  there are exactly  $d$ (different) solutions in $\Fq.$
\end{lemma}

\begin{lemma}\cite[Lemma 1]{Redei1950} \label{L_Binomial}
Let $m, n \geq 1$ be integers and $a, b \in \mathbb{F}_q^*$. Then $\gcd\left(x^n-a, x^m-b\right)$ has degree 0 or $d:=\gcd(m, n)$ (over an arbitrary field). Moreover, $\gcd\left(x^n-a, x^m-b\right)$ has degree $d$ if and only if $a^{m / d}=b^{n / d}$.
\end{lemma}

In \cite[Theorem 4.1]{Aydin2017}, the authors proved that when $\gcd(n, p) = \gcd(m, p) = 1$, the polynomial $\gcd(x^n - a, x^m - b)$ is either $1$ or $x^{\gcd(n, m)} - c$, for some $c \in \mathbb{F}_q$. Following their proof, we extend the result to the general case in the following lemma, allowing also $\gcd(n, p) \neq 1$ or $\gcd(m, p) \neq  1.$

\begin{lemma}\label{L.2}
Let $f = x^n - a$ and $g = x^m - b$ be two polynomials in $\mathbb{F}_q[x]$. %where $q$ is a prime power of $p$. 
Then either $\gcd(f, g)=1$ or  $$ \gcd(f, g) = x^{\gcd(m,n)} -  a^u b^v ,$$  
where $ u $ and $ v $ are integers such that $ \gcd(m, n) = u n + v m $.
\end{lemma}

\begin{proof}
  Let $n= p^r n'$ and $m= p^s m'$ such that $\gcd(n', p) = \gcd(m', p) = 1$.  Then there exists a unique $a'=a^{p^{-r}} \in \mathbb{F}_q^{*}$ and $b'=b^{p^{-s}} \in \mathbb{F}_q^{*}$ such that    
$$f = x^n - a  = \left(x^{n'} - a'\right)^{p^r}, \quad \text{and} \ \ g = x^m - b  = \left(x^{m'} - b'\right)^{p^s}.$$
 Let $t = \operatorname{ord}_{n' r_1 m' r_2}(q)$ be the order of $q$ modulo $n' r_1 m' r_2,$  where $r_1$ and $r_2$ are the multiplicative orders of $a$ and $b$ in $\mathbb{F}_q^{*}$, respectively. Then the field extension $\mathbb{F}_{q^t}$ contains both the roots of $x^{n'} - a'$ and $x^{m'} - b'$. 
 
Let $\zeta$ be a primitive $(n' m')^{\text{th}}$ root of unity. Then $\zeta^{n'}$ and $\zeta^{m'}$ are respectively the $m'^{\text{th}}$ and $n'^{\text{th}}$ roots of unity. If $\gcd(x^{n'} - a', x^{m'} - b') = 1,$ then $\gcd(f, g) = 1$, and we are done.

Now, suppose there exists a common root $\delta$ of $x^{n'} - a'$ and $x^{m'} - b'$, which is an $n'^{\text{th}}$ root of $a'$ and an $m'^{\text{th}}$ root of $b'$. As in the proof of \cite[Theorem 4.1]{Aydin2017}, the roots of $x^{n'} - a'$ and $x^{m'} - b'$ are, respectively,
$$ \delta, \delta \zeta^{m'}, \delta (\zeta^{m'})^2, \ldots, \delta (\zeta^{m'})^{n'-1}, \ \   \text{and} \ \ \delta, \delta \zeta^{n'}, \delta (\zeta^{n'})^2, \ldots, \delta (\zeta^{n'})^{m'-1} .  $$
Let $d' := \gcd(n', m')$ and $\beta := \zeta^{\lcm{(n', m')}},$ then the roots of $\gcd(x^{n'} - a', x^{m'} - b')$ are 
$$   \delta, \delta \beta, \delta \beta^2, \ldots, \delta \beta^{d' - 1},  $$
and so  
$$ \gcd(x^{n'} - a', x^{m'} - b') = x^{d'} - \delta^{d'} .$$

Note that the roots of $f$  are: 
$$ \overbrace{\delta,  \ldots, \delta}^{p^r \text{ times}}, \overbrace{\delta \zeta^{m'},  \ldots, \delta \zeta^{m'}}^{p^r \text{ times}}, \ldots,  \overbrace{\delta (\zeta^{m'})^{n'-1},  \ldots, \delta (\zeta^{m'})^{n'-1}}^{p^r \text{ times}},  $$
and the roots of $g$ are
$$ \overbrace{\delta,  \ldots, \delta}^{p^s \text{ times}}, \overbrace{\delta \zeta^{n'},  \ldots, \delta \zeta^{n'}}^{p^s \text{ times}}, \ldots,  \overbrace{\delta (\zeta^{n'})^{m'-1},  \ldots, \delta (\zeta^{n'})^{m'-1}}^{p^s \text{ times}}.  $$ 
It follows that the roots of $ \gcd(f, g)$ are
$$ \overbrace{\delta,  \ldots, \delta}^{p^{\min(r, s)}  \text{ times}}, \overbrace{\delta \beta,  \ldots, \delta \beta}^{p^{\min(r, s)} \text{ times}}, \ldots,  
\overbrace{\delta \beta^{d' - 1},  \ldots, \delta \beta^{d' - 1}}^{p^{\min(r, s)} \text{ times}}.  $$
Therefore, $\deg(\gcd(f, g)) = p^{\min(r, s)} d',$ and 
$$   \gcd(f, g) = \left(x^{d'} - \delta^{d'}\right)^{p^{\min(r, s)}} =  x^{p^{\min(r, s)} d'} - \delta^{d' p^{\min(r, s)}}=x^{\gcd(m,n)} - \delta^{\gcd(m,n)}.$$
To complete the proof we need to show that $\delta^d= a^u  b^v\in \Fq, $ for some integers   $u, v$ such that $d= \gcd(m,n)=u n+v m$.
We know that $\delta^{n}=a$ and $\delta^{m}=b$, which implies
$$
\delta^d=\delta^{u n+v m}=a^u  b^v \in \mathbb{F}_q.
$$
Finally, we proved that $ \gcd(f, g) = x^{d} -  a^u b^v $  for some integers $ u $ and $ v $ such that $ d = \gcd(m, n) = u n + v m $.
\qed
\end{proof}

More generally, we prove by induction the following result:
\begin{lemma}\label{Iso_Genaral}
 Let $m$ be an integer and let $f_i = x^{n_i} - a_i$ for $i = 1, \ldots, m$ be polynomials in $\mathbb{F}_q[x]$.
 Then $ \gcd(f_1(x), f_2(x), \ldots, f_m(x))$ is either $1$ or of the form $$x^{d} - \dprod{i=1}{m} a_i^{u_i},$$ where the $u_i$ are integers such that $d = \gcd(n_1, n_2, \ldots, n_m) = \dsum{i=1}{m} u_i n_i$.
\end{lemma}
\begin{proof}
Let $f_i = x^{n_i} - a_i$ for $i=1, \ldots, m$. We will use induction to prove that $\gcd(f_1, f_2, \ldots, f_m)$ is either $1$ or of the form $x^d - \dprod{i=1}{m} a_i^{u_i}$, where $d = \gcd(n_1, n_2, \ldots, n_m)$. 

\noindent For $m=2,$ the result is verified by Lemma \ref{L.2}. So, let us assume that the result is valid for the $m-1$ polynomials $f_1, \ldots, f_{m-1}$. That is, 
$$
\gcd(f_1, f_2, \ldots, f_{m-1}) = 
\begin{cases} 
x^{d'} - \dprod{i=1}{m-1} a_i^{u_i}, & \text{if } d' = \gcd(n_1, n_2, \ldots, n_{m-1}) = \dsum{i=1}{m-1} u_i n_i  \\
1, & \text{else}
\end{cases}.
$$ 

\noindent Let us now consider $f_1, f_2, \ldots, f_m$, and put  $g := \gcd(f_1, f_2, \ldots, f_{m-1})$. 
%Then, by the inductive hypothesis, $g = 1$ or $g = x^{d'} - \dprod{i=1}{m-1} a_i^{u_i}$, where $d' = \gcd(n_1, n_2, \ldots, n_{m-1})$. 
If $g = 1$, then $\gcd(g, f_m) = 1$.  Else, $g = x^{d'} - \dprod{i=1}{m-1} a_i^{u_i}$, then by applying   Lemma \ref{L.2} to $g$ and $f_m = x^{n_m} - a_m$, we obtain
$$
\gcd(g, f_m) = 
\begin{cases} 
x^d - \dprod{i=1}{m} a_i^{u_i}, & \text{if there is a common root,}\\
1, & \text{if there is no common root of } g \text{ and } f_m,
\end{cases}
$$
where $d = \gcd(d', n_m) = \gcd(n_1, n_2, \ldots, n_m)$, and $u_i$ are integers such that $d = \dsum{i=1}{m} u_i n_i$. Then the result holds.
\qed
\end{proof}

\section{Equivalence of Trinomial Code Families}\label{sec:eq_trinomial}

 In this section, we study the equivalence between $\ell$-trinomial code families.

 \begin{definition}\label{Def_Iso}
Let $a_0, a_{\ell}, b_0, b_{\ell}$ be nonzero elements of $\mathbb{F}_q$, and let $\ell$ be an integer such that $0 < \ell < n$. We say that $(a_0, a_{\ell})$ and $(b_0, b_{\ell})$ are \textit{$(n,\ell)$-equivalent} in $\mathbb{F}_q^* \times \mathbb{F}_q^*$, denoted by 
$$
(a_0, a_{\ell}) \sim_{(n,\ell)} (b_0, b_{\ell}),
$$
if there exists an $\alpha \in \mathbb{F}_q^*$ such that the following map
  \begin{equation}
  	\begin{array}{cccc}
  	\varphi_{\alpha}:& \mathbb{F}_q[x] /\langle x^n-b_{\ell}x^{\ell}-b_0 \rangle  &\longrightarrow & \mathbb{F}_q[x] /\langle x^n-a_{\ell} x^{\ell}-a_0 \rangle, \\ 
  	& f(x) & \longmapsto &  f(\alpha x),
  	\end{array}
  	\end{equation} 
is an $\Fq$-algebra isomorphism. Note that such an isomorphism preserves the Hamming weight, i.e., 
$$
d(\varphi_{\alpha}(f(x)), \varphi_{\alpha}(g(x)))=d(f(x), g(x)),$$
for all  $f(x), g(x) \in \Fq[x] /\langle x^n-b_{\ell}x^{\ell}-b_0 \rangle.$
\end{definition}

\begin{remark}\label{Remark.3}
\begin{enumerate}
    \item  For any  integer  $0<\ell<n,$ the relation $\sim_{(n,\ell)}$ is an equivalence relation on $\Fq^*\times \Fq^*.$
 \item  Note that the $(n,\ell)$-equivalence relation in the above definition generalizes the $n$-equivalence of constacyclic codes studied in \cite[Definition 3.1]{Chen2014}, which was denoted by $\lambda \sim_n \mu$.

\end{enumerate}
\end{remark}

In the following theorem we give essential characterizations of the $(n,\ell)$-equivalence  between two classes of $\ell$-trinomial codes of length $n$ over $\mathbb{F}_q$. In the statement we will use the component-wise product of two length $n$ vectors $x$ and $y$, also known as the \textit{Schur product},  defined as
$$
(x_0, x_1, \ldots, x_{n-1}) \star (y_0, y_1, \ldots, y_{n-1}):=(x_0 y_0, x_1 y_1, \ldots, x_{n-1} y_{n-1}) .
$$

\begin{theorem}\label{Th.1}
Let $0<\ell<n$ be an integer,  $(a_0, a_{\ell})$ and $(b_0, b_{\ell})$ be elements of $\mathbb{F}_q^{*} \times \mathbb{F}_q^{*}$,  and  $\xi$ be  a primitive element of $\Fq$. 
The following statements are equivalent:

\begin{enumerate}
    \item $(a_0, a_{\ell}) \sim_{(n,\ell)} (b_0, b_{\ell}).$

    \item The polynomials $ a_i x^{n-i } - b_i \in \mathbb{F}_q[x]$, with $i \in \{0, \ell\}$, have a common root in $\mathbb{F}_q^*.$

    \item The polynomial $\gcd(a_0 x^n - b_0, a_{\ell} x^{n - \ell} - b_{\ell})$ has at least one root in $\mathbb{F}_q^*.$
       \item The polynomial $\gcd(x^n - b_0 a_0^{-1}, x^{n - \ell} - b_{\ell} a_{\ell}^{-1})$ has at least one root in $\mathbb{F}_q^*.$
     \item  There exists $\alpha\in \Fq^*$ such that $ (a_0,a_{\ell})\star ( \alpha^n,   \alpha^{n-\ell})= (b_0,b_{\ell}) $.
     %,$ where $ \star$ is the component-wise product on $ \Fq^*\times \Fq^*. $ 
     
      \item  $ (a_0, a_{\ell})^{-1} \star(b_0,b_{\ell})\in H,$ where   $H$ is the cyclic subgroup of $\Fq^*\times \Fq^*$ generated by $(\xi^n,\xi^{n-\ell}) .$
\end{enumerate}
The equivalence between (1) and (6) implies that the number of $(n,\ell)$-equivalence classes is $$N_{(n,\ell)}:= \dfrac{ (q-1)^2}{ \lcm(\frac{q-1}{ \gcd(n,q-1)}, \frac{q-1}{ \gcd(n-\ell,q-1)})} = (q-1) \gcd(q-1,n,\ell).
%(q-1)\gcd\left( \frac{q-1}{\gcd(n,q-1)}, \frac{q-1}{\gcd(n-\ell ,q-1)} \right).
$$
\end{theorem}

\begin{proof}
\begin{enumerate}
    \item[(1) $\Rightarrow$ (2)] 
    Suppose that $(a_0, a_{\ell}) \sim_{(n,\ell)} (b_0, b_{\ell})$. Then -- by Definition \ref{Def_Iso} -- there exists $\alpha \in \mathbb{F}_q^*$ such that the map
    $$
    \varphi_{\alpha} : \mathbb{F}_q[x] /\langle x^n - b_{\ell} x^{\ell} - b_0 \rangle \to \mathbb{F}_q[x] /\langle x^n - a_{\ell} x^{\ell} - a_0 \rangle, \quad f(x) \ \mapsto \ f(\alpha x)
    $$
    is an $\mathbb{F}_q$-algebra isometry. It follows that

    $$
    \varphi_{\alpha}(x^i) = \varphi_{\alpha}(x)^i = \alpha^i x^i \ \mod (x^n - a_{\ell} x^{\ell} - a_0), \ \forall i = 0, 1, \ldots, n-1.
    $$

    Since $\varphi_{\alpha}$ is an $\mathbb{F}_q$-algebra isometry and $\varphi(x^n - b_{\ell} x^{\ell} - b_0) = 0 \ \mod (x^n - a_{\ell} x^{\ell} - a_0)$, then
    \begin{equation}
    \varphi_{\alpha}(x^n) = b_{\ell} \alpha^{\ell} x^{\ell} + b_0.
    \end{equation}

    On the other hand,
    \begin{equation}
    \varphi_{\alpha}(x^n) = \alpha^n x^n = \alpha^n (a_{\ell} x^{\ell} + a_0) = \alpha^n a_{\ell} x^{\ell} + \alpha^n a_0.
    \end{equation}

    Comparing term by term, we deduce that $a_0 \alpha^n = b_0$ and $a_{\ell} \alpha^{n-\ell} = b_{\ell}$, which means that $\alpha$ is a common root of the polynomials $a_0 x^n - b_0$ and $a_{\ell} x^{n-\ell} - b_{\ell}$.

    \item[(2) $\Rightarrow$ (3)] and (3) $\Rightarrow$ (4) are immediate.
    \item[(4) $\Rightarrow$ (5)]
    Let $\alpha $ be a root of the polynomial $\gcd(x^{n} - b_0 a_0^{-1}, x^{n-\ell} - b_{\ell} a_{\ell}^{-1})$. Then $\alpha$ is a common root of the polynomials $a_0 x^n - b_0$ and $a_{\ell} x^{n-\ell} - b_{\ell}$. It follows that  $a_i \alpha^{n-i}= b_i,\ \text{for any} \ i\in \{ 0,\ell\}$, i.e.,
$$  (b_0, b_{\ell} )= (\alpha^n a_0, \alpha^{n-\ell} a_{\ell})=(\alpha^n, \alpha^{n-\ell})\star(a_0, a_{\ell} ).$$ 
 \item[(5) $\Rightarrow$ (6)]
Suppose that there is $\alpha\in \Fq^*$ such that $ (b_0, b_{\ell} ) =(\alpha^n, \alpha^{n-\ell})\star(a_0, a_{\ell} ). $ Then,  
$$  (a_0,a_{\ell})^{-1}\star(b_0, b_{\ell})= (a_0^{-1}b_0, a_{\ell}^{-1}b_{\ell} ) = (\alpha^{n}, \alpha^{n-\ell} )= (\xi^{j n}, \xi^{ j(n-\ell)} )= (\xi^{n}, \xi^{(n-\ell)})^j .$$
It follows that  $ (a_0,a_{\ell})^{-1}\star(b_0, b_{\ell}) $ belongs to the cyclic subgroup  $H$ generated by $  (\xi^{n}, \xi^{(n-\ell)}) $  as a subgroup of $ \Fq^{*}\times \Fq^{*}.$

    \item[(6) $\Rightarrow$ (1)]  
     Suppose that $ (a_0,a_{\ell})^{-1}\star(b_0, b_{\ell}) $ is an element of the cyclic subgroup  $H$ generated by $  (\xi^{n}, \xi^{(n-\ell)})  $  as a subgroup of $ \Fq^{*}\times \Fq^{*}$. Then there exists an integer $ h$ such that 
     $$ (a_0,a_{\ell})^{-1}\star(b_0, b_{\ell})=(\xi^n,\xi^{n-\ell})^h=( \xi^{hn}, \xi^{h(n-\ell)}) $$
    For $ \beta= \xi^{h}, $ we obtain  that $  a_i \beta^{n-i}= b_i, $ for any $ i\in \{ 0,k\}.$
Now, let consider the map $ \tilde{\varphi}_{\beta}, $ as follows: 
  \begin{equation}
        \begin{array}{cccc}
        \tilde{\varphi}_{\beta} : & \mathbb{F}_q[x]  & \longrightarrow & \mathbb{F}_q[x] /\langle x^n - a_{\ell} x^{\ell} - a_0 \rangle, \\ 
        & f(x) & \longmapsto & f(\beta x).
        \end{array}
    \end{equation} 
	$ \tilde{\varphi}_{\beta}$ is a surjective $\Fq$-algebra homomorphism, indeed, for  all $0\leq j \leq n-1$, $x^{j}  = \tilde{\varphi}_{\beta}(\beta^{-j}x^{j})$. 
	Moreover, 
 $$ 
    \begin{array}{rl}
    \tilde{\varphi}_{\beta} (x^n - b_{\ell} x^{\ell} - b_0) & = \beta^n x^n - \beta^{\ell} b_{\ell} x^{\ell} - b_0 \\
    & = \beta^n x^n - \beta^n a_{\ell} x^{\ell} - \beta^n a_0 \\
    & = \beta^n(x^n - a_{\ell} x^{\ell} - a_0) \\
    & = 0 \mod (x^n - a_{\ell} x^{\ell} - a_0).
    \end{array}
    $$ 
	So $\langle x^n - b_{\ell} x^{\ell} - b_0\rangle \subseteq \ker\tilde{\varphi}_{\beta}$. And for all $f(x)\in \ker\tilde{\varphi}_{\beta}$,
    $$f(\beta x)= 0~ (\text{mod } x^{n}- a_{\ell} x^{\ell} - a_0),$$ then there exists $g(x)\in \Fq[x]$ such that $f(\beta x)=g(x)(x^{n}- a_{\ell} x^{\ell} - a_0)$, thus 
 $$
  \begin{array}{rl}
 f(x)&=g(\beta^{-1}x)(\beta^{-n}x^{n}- \beta^{-\ell} a_{\ell} x^{\ell} - a_0) \\
 & =\beta^{-n}g(\beta^{-1}x)(x^{n}- \beta^{n-\ell}a_{\ell} x^{\ell} - \beta^{n} a_0)\\
 &= \beta^{-n}g(\beta^{-1}x)( x^n-  b_{\ell} x^{\ell} - b_0), \ \text{since  $a_i \beta^{n-i} =b_i, \ \forall i \in \{0, \ell\}$} \\ 
 \end{array}$$
	So $\ker\tilde{\varphi}_{\beta} \subseteq \langle x^{n}- a_{\ell} x^{\ell} - a_0\rangle $ and hence $ker\tilde{\varphi}_{\beta} = \langle x^{n}-  b_{\ell} x^{\ell} - b_0 \rangle$.
	Therefore  by the  first isomorphism  theorem the map 
      \begin{equation}
        \begin{array}{cccc}
        \varphi_{\beta} : & \mathbb{F}_q[x]/\langle x^n - b_{\ell} x^{\ell} - b_0 \rangle  & \longrightarrow & \mathbb{F}_q[x] /\langle x^n - a_{\ell} x^{\ell} - a_0 \rangle, \\ 
        & f(x) & \longmapsto & f(\beta x).
        \end{array}
    \end{equation} 
	is an $\Fq$-algebra isomorphism. As  the weights of $f(x)$ and $f(\beta x)$ are the same, the result holds. 
\end{enumerate}

By the equivalence between (1) and (6), we deduce that the number of $(n,\ell)$-equivalence classes on $\Fq^*\times \Fq^*$ corresponds to the order of the group $ \left( \Fq^*\times \Fq^* \right)/H,$ which equals $ N_{(n,\ell)}=\dfrac{ (q-1)^2}{ \lcm(\frac{q-1}{ \gcd(n,q-1)}, \frac{q-1}{ \gcd(n-\ell,q-1)})}
%(q-1)\gcd\left( \frac{q-1}{\gcd(n,q-1)}, \frac{q-1}{\gcd(n-\ell ,q-1)} \right).
$, which in turn is equal to
$$    
\begin{array}{rl}
  &\dfrac{(q-1)^2}{\dfrac{(q-1)^2}{ \gcd(n,q-1)\gcd(n-\ell,q-1)} \dfrac{1}{\gcd(\dfrac{q-1}{ \gcd(n,q-1)}, \dfrac{q-1}{ \gcd(n-\ell,q-1)})}} \\
    
     &=\gcd(n,q-1)\gcd(n-\ell,q-1)\gcd\left(\dfrac{q-1}{ \gcd(n,q-1)}, \dfrac{q-1}{ \gcd(n-\ell,q-1)}\right)\\
    &= \dfrac{\gcd(n,q-1)\gcd(n-\ell,q-1)(q-1)}{\gcd(n,q-1)\gcd(n-\ell,q-1)} \gcd\big(\gcd(n,q-1),\gcd(n-\ell,q-1)\big)\\
    &= (q-1) \gcd\big(\gcd(n,q-1),\gcd(n-\ell,q-1)\big)\\
    &= (q-1) \gcd(q-1,n,\ell) .
 \end{array}
$$
\qed
\end{proof}

Using the equivalence between the assertions $(1)$ and $(5)$ of Theorem \ref{Th.1} we derive a characterization regarding the equivalence between the class of $\ell$-trinomial codes  associated with  $x^n - a_{\ell}x^{\ell} - a_0$ and the class  associated with $x^n - x^{\ell} - 1$ in the following.

For this we first derive a result from Lemmas \ref{L_Binomial} and \ref{L.2}, which we will then use in the proof of Corollary \ref{Cor.ver2}.

\begin{lemma}\label{Cor_Binomial}
Let $m,n$ be two positive integers and $u,v \in \mathbb Z$ be such that $\gcd(m, n)= un+vm$. Let furthermore $f=x^n-a$ and $\ g=x^m-b$ two polynomials in $\Fq[x]$. 
Then  $$ \gcd(f, g) =
\begin{cases}
x^{\gcd(n,m)} - a^{u}b^{v}  & \text{ if }  a^{ \frac{m}{\gcd(m, n)}}=b^{\frac{n}{\gcd(m, n)}} \\
1 & \text{ else }
\end{cases} .$$ 
\end{lemma}

\begin{proof}

%Let $f = x^{n} - a$ and $g = x^{m} - b$ be two polynomials in $\mathbb{F}_{q}[x]$. 
By Lemma~\ref{L_Binomial} we know that $\gcd(f,g)$ has degree $0$ or $d := \gcd(m,n)$, and it has degree $d>0$ if and only if $a^{m/d} =b^{n/d}$.  
If $\gcd(f,g)$ has degree $d> 0,$ then  $\gcd(f,g)\neq 1,$ and by Lemma \ref{L.2},  we have   $ \gcd(f, g) = x^{\gcd(m,n)} -  a^u b^v ,$  for $ u $ and $ v $ with $ \gcd(m, n) = u n + v m $.
%Hence,  by the second part of  Lemma~\ref{L_Binomial}, we obtain that  $\gcd(f, g) = x^{\gcd(m,n)} -  a^u b^v$ if and only if $a^{\frac{m}{\gcd(m,n)}}  = b^{\frac{n}{\gcd(m,n)}}$. 
%Finally, we obtain that
% $$ \gcd(f, g) =
%\begin{cases}
%x^{\gcd(n,m)} - a^{u}b^{v}  & \text{ if }  a^{ \frac{m}{\gcd(m, n)}}=b^{\frac{n}{\gcd(m, n)}} \\
%1 & \text{ else }
%\end{cases} .$$ 

\end{proof}

\begin{corollary}\label{Cor.ver2}
Let $\ell$ be an integer such that  $0<\ell<n$,  and  $(a_0, a_{\ell})$ be an element of $\mathbb{F}_q^{*} \times \mathbb{F}_q^{*}$.  Then the following statements are equivalent.
    \begin{enumerate}
        \item  The class of $\ell$-trinomial codes  associated with  $x^n - a_{\ell}x^{\ell} - a_0$ is equivalent to the class of $\ell$-trinomial codes associated with $x^n - x^{\ell} - 1.$
        \item  There exists $ \alpha \in \Fq^*$ such that $(a_0, a_{\ell})\star(\alpha^n, \alpha^{n-\ell}) = (1, 1)$. 
        \item  There exists $   \alpha \in \Fq^* $ being an $n$-th root of $a_0$  such that $   a_0 =  \alpha^\ell a_\ell $.
        \item  We have $$ a_{\ell}^{\frac{n}{\gcd(n,n-\ell)}}=a_0^{\frac{n-\ell}{{\gcd(n,n-\ell)}}}$$ and $x^{ \gcd(n,n-\ell)}-a_0^v a_{\ell}^u$ has a root, for $u,v\in \mathbb Z$ with $ \gcd(n,n-\ell)=nv+u(n-\ell) .$ 
        %\textcolor{purple}{The $b_\ell$ should be $a_\ell$ (also below).}
        \end{enumerate}
\end{corollary}
\begin{proof}
\begin{itemize}
    \item[(1) $\Rightarrow$ (2)] Follows from the equivalence of assertions (1) and (5) of Theorem \ref{Th.1}.
 \item[(2) $\Rightarrow$ (3)]  
 %Let $a_0, a_{\ell} \in \Fq^*$ such that the class of  $\ell$-trinomial codes  associated with $x^n - a_{\ell}x^{\ell} - a_0$ is equivalent to the class of trinomial codes associated with $x^n - x^{\ell} - 1$.  By Theorem \ref{Th.1}, there exist $\beta \in \Fq^*$ such that  $ a_0 \beta^n=1, $  and  $ a_{\ell} \beta^{n-\ell}=1.   $  It follows that $\alpha:=\beta^{-1}$ is a $n$-th root of $a_0$ such that $  a_0 \alpha^{-n}=  a_{\ell} \alpha^{\ell-n} $ and so $  a_{\ell}^{-1}a_0 =  \alpha^{\ell}.$
 If $(a_0, a_{\ell})\star(\alpha^n, \alpha^{n-\ell}) = (1, 1)$, then 
 $$a_0 \alpha^n=1 \quad \text{ and } \quad a_\ell \alpha^{n-\ell}=1$$
 $$\iff a_0 =\alpha^{-n} \quad \text{ and } \quad a_\ell \alpha^{-\ell}=a_0$$
  i.e., $\beta:=\alpha^{-1}$ is an $n$-th root of $a_0$ and $a_0=\beta^\ell a_\ell$.

 \item[(3) $\Rightarrow$ (4)] Suppose that $\alpha $ is an $n$-th root of $a_0$ such that $ a_{\ell}^{-1}  a_0 = \alpha^\ell$.
 %then $$ \alpha^n = a_0 \ \ \text{ and } \ \  a_{\ell}^{-1}  a_0 = \alpha^\ell.$$ 
 It follows that 
 $$ \alpha^{n} a_0^{-1} = 1 \ \ \text{ and } \ \  a_{\ell}^{-1}  a_0= a_{\ell}^{-1}  a_0  \alpha^{n} a_0^{-1} = \alpha^{\ell}.$$ 
 Hence
 %$$ \alpha^{n} a_0^{-1} = 1 \ \ \text{ and } \ \  a_{\ell}^{-1}    \alpha^{n}  = \alpha^{\ell}.$$ 
% And so,
 $$ \alpha^{n}  = a_0\ \ \text{ and } \ \    \alpha^{n-\ell}  = a_{\ell},$$
 which means that  $\alpha $ is a common root of $ x^n-a_0$ and $ x^{n-\ell}-a_{\ell}$. 
Moreover, we get
 $$ a_{\ell}^{\frac{n}{\gcd(n,n-\ell)}}= \alpha^{\frac{(n-\ell)n}{\gcd(n,n-\ell)}}=a_0^{\frac{n-\ell}{\gcd(n,n-\ell)}},$$
 i.e., we can use Lemma \ref{Cor_Binomial} to deduce that
 $ \gcd( x^n-a_0, x^{n-\ell}-a_{\ell}) = x^{\gcd(n, n-\ell)} - a_{\ell}^u a_0^v ,  $     for $u, v\in \mathbb Z$ with $  vn+u(n-\ell)=\gcd(n,n-\ell).$   It follows that $\alpha$ is a root of $x^{\gcd(n, n-\ell)} - a_{\ell}^u a_0^v$.

  \item[(4) $\Rightarrow$ (1)] Suppose that 
  $ a_{\ell}^{\frac{n}{\gcd(n,n-\ell)}}=a_0^{\frac{n-\ell}{{\gcd(n,n-\ell)}}}$ and that $\alpha$ is a root of $x^{ \gcd(n,n-\ell)}-a_0^v a_{\ell}^u$. By Lemma \ref{Cor_Binomial} it follows that $x^{ \gcd(n,n-\ell)}-a_0^v a_{\ell}^u = \gcd(x^n-a_0, x^{n-\ell}-a_\ell)$ and hence that 
  $ \alpha $ is a common root of $ x^n-a_0 $ and $ x^{n-\ell}-a_{\ell.}$ The statement now follows from  Theorem \ref{Th.1}. 
  % we can verify  that 
%     \begin{equation}
%   	\begin{array}{cccc}
%   	\varphi_{\alpha}:& \Fq[x] /\langle x^n-a_{\ell}x^{\ell}-a_0 \rangle  &\longrightarrow & \Fq[x] /\langle x^n-x^{\ell}-1 \rangle, \\ & & & \\
%   	& f(x) & \longmapsto &  f(\alpha x).
%   	\end{array}
%   	\end{equation} 
% is an $\Fq$-algebra isomorphism, and so we conclude that the class of $\ell$-trinomial codes associated with $x^n - a_{\ell}x^{\ell} - a_0$ is equivalent to the class of $\ell$-trinomial codes associated with the polynomial $f(x) = x^n - x^{\ell} - 1$.
 \end{itemize}
 \qed
 \end{proof}

Once we know that two classes are equivalent, we can determine how many different $\varphi_\alpha$ are isometric isomorphisms from one class to the other:
\begin{corollary}
With the same notation as in Theorem \ref{Th.1} the number of isometric $\Fq$-algebra isomorphisms $\varphi_{\alpha}$ between $ \Fq[x]/\langle x^n-b_{\ell} x^{\ell}-b_0\rangle $ and $  \Fq[x]/\langle x^n-a_{\ell} x^{\ell}-a_0\rangle$ is $ \gcd(n, n-\ell, q-1) $, if $(a_0, a_{\ell}) \sim_{(n,\ell)} (b_0, b_{\ell})$ (otherwise, there are none).
\end{corollary}

\begin{proof}
By the first step in the proof of Theorem \ref{Th.1}, $ (a_0, a_{\ell}) \sim_{(n,\ell)} (b_0, b_{\ell}) $ under $\varphi_\alpha$ if and only if $\alpha$ is a root of the polynomial $ \gcd(x^{n} - b_0 a_0^{-1}, x^{n-\ell} - b_{\ell} a_{\ell}^{-1}) $. By Lemma \ref{L.2}, we obtain that
$$ \gcd(x^{n} - b_0 a_0^{-1}, x^{n-\ell} - b_{\ell} a_{\ell}^{-1}) = x^{\gcd(n,n-\ell)} -  (b_0 a_0^{-1})^u (b_{\ell} a_{\ell}^{-1})^v,$$ 
where $ d = \gcd(n, n-\ell)=un+v(n-\ell) $  for some integer $u$ and $v$.  By Lemma \ref{L_soltion}, the number of roots of $x^{\gcd(n,n-\ell)} - (b_0 a_0^{-1})^u (b_{\ell} a_{\ell}^{-1})^v$ in $ \Fq$ is equal to $ \gcd(n, n-\ell, q-1) $, hence there are  $ \gcd(n, n-\ell, q-1) $ many different $\alpha$ such that $\varphi_\alpha$ is an isomorphism. (For the last step note that $\alpha \in \Fq^*$ since $0$ is not a root of the polynomials above.)
\qed
\end{proof}

  In the following result, we describe the associated polynomials of possible $\ell$-trinomial codes based on the $(n,\ell)$-equivalence relation defined above.
\begin{theorem}
     \label{Equiv_2}
 Let $n, \ell$ be two integers such that $0<\ell<n$ and let $\xi$ be  a primitive element of $\Fq$.  Set  $ d:=\gcd \left(\frac{q-1}{ \gcd(n,q-1)}, \frac{q-1}{ \gcd(n-\ell,q-1)}\right)$  and $ d_i:=\gcd(n-i,q-1)$, for $ \ i\in \{0,\ell\}$. Moreover, let $a_0,a_\ell \in \Fq^*$. 
 \begin{enumerate}
     \item  If $d=1,$ then the class of $\ell$-trinomial codes associated with $x^n-a_{\ell}x^{\ell}-a_0$ is equivalent to the class of $ \ell$-trinomial codes associated with $x^n-\xi^j x^{\ell}- \xi^i,$ for some $i\in \{0,1,\ldots, d_0-1 \}$ and $ j\in \{0,1,\ldots, d_{\ell}-1\}$. 
    \item    If $d\neq 1,$ then the class of $\ell$-trinomial codes associated with $x^n-a_{\ell}x^{\ell}-a_0$ is equivalent to the class of $ \ell$-trinomial codes associated with  $x^n-\xi^{j} x^{\ell}- \xi^{i+hn},$ for some $i\in\{0,1,\ldots, d_0-1\}$,  $j\in \{0,1,\ldots, d_{\ell}-1\}$  and $ h\in \{0,\ldots, d-1\}$.  
 \end{enumerate}
\end{theorem}

\begin{proof}
\begin{enumerate}
    \item  If $ d = 1 $, then the cyclic group $ H $ generated by $ (\xi^n, \xi^{n-\ell}) $ is isomorphic to the group $ \langle \xi^n \rangle \times \langle \xi^{n-\ell} \rangle $ and has order $ \frac{(q-1)^2}{d_0 d_{\ell}} $. By Theorem \ref{Th.1}, the number of $ (n, \ell) $-equivalence classes is $ d_0 d_{\ell} $. Therefore, we can partition $ \Fq^* \times \Fq^* $ as 
$$
\Fq^* \times \Fq^* =  \bigcup_{i=0}^{d_0-1} \bigcup_{j=0}^{d_{\ell}-1} (\xi^i, \xi^j) H.
$$
Then any pair $(a_0,a_{\ell})$ is $(n,\ell)-$equivalent to one of the pairs   $(\xi^i, \xi^j), $ for  $i=0,1,\ldots, d_0-1$, and  $j=0,1,\ldots, d_{\ell}-1$. 
%It follows that the class of $\ell$-trinomial codes associated with $x^n-a_{\ell}x^{\ell}-a_0$ is equivalent to the class of $ \ell$-trinomial codes associated with $x^n-\xi^j x^{\ell}- \xi^i,$ for $i=0,1,\ldots, d_0-1 $ and $ j=0,1,\ldots, d_{\ell}-1.$

 \item 
 If $d\neq 1, $ the number of  $(n,\ell)$-equivalence classes is $ d d_0 d_{\ell} ,$ and so we partition $ \Fq^* \times \Fq^* $ as 
$$  \Fq^{*}\times \Fq^{*} = \bigcup_{h=0}^{d-1}  \bigcup_{i=0}^{d_0-1} \bigcup_{j=0}^{d_{\ell}-1} (\xi^{i+hn} ,\xi^{j}) H , $$
which implies the second statement, similarly to the first case.
\end{enumerate}
  \qed
\end{proof}

In the following, we deduce additional properties of the $(n,\ell)$-equivalence between families of $\ell$-trinomial codes, linking this $(n,\ell)$-equivalence to that of constacyclic codes from \cite{Chen2014,Chen2012}.

\begin{corollary}
    \label{Eq_Consta}
Let $ (a_0, a_{\ell}) $ and $ (b_0, b_{\ell}) $ be elements of $ \Fq^{*} \times \Fq^{*} $ such that $ (a_0, a_{\ell}) \sim_{(n, \ell)} (b_0, b_{\ell}) $. Then, for each $ i \in \{0, \ell\} $,
\begin{enumerate}
    \item $ a_i^{-1} b_i \in \langle \xi^{n-i} \rangle $, where $ \xi $ is a primitive element of $ \Fq^* $,
    \item $ (a_i^{-1} b_i)^{d_i} = 1 $, where $ d_i = \frac{q-1}{\gcd(n-i, q-1)} $,
    \item $ a_i \sim_{n-i} b_i $, i.e., the class of $a_i$-constacyclic codes of length $n-i$ is equivalent to the class of $b_i$-constacyclic codes of length $n-i$ over $\Fq.$
\end{enumerate}
\end{corollary}

This last result implies that we can use known results about the equivalence of families of constacyclic codes for the equivalence of trinomial codes.

\section{Equivalence of $p^{\ell}$-trinomial Codes of Length $n=p^{\ell+r}$.}\label{sec:eq_trinomial_special}

In this section we study $p^{\ell}$-trinomial  codes of length $n = p^{\ell+r}$, where $p$ is the characteristic of the underlying field $\Fq$ and $r$ an integer. First, we recall the following lemma, which combines Artin-Schreier's theorem \cite[Theorem 12.2.1]{Roman1995} and \cite[Corollary 3.79]{Lidl1987}.

\begin{lemma}\label{LLL.6}
 Let $a \in \mathbb{F}_q$ and let $p$ be the characteristic of $\mathbb{F}_q$. Then the trinomial $x^p - x - a$ is irreducible in $\mathbb{F}_q[x]$ if and only if $\Tr_{\mathbb{F}_q / \mathbb{F}_p }(a) \neq 0$; otherwise it splits (into linear factors) in $\mathbb{F}_q$. If it splits, then the roots are of the form $\beta+i$ for $i=0,1,\dots,p-1$ and any root $\beta$.

\end{lemma}

From the above lemma we easily deduce the following proposition.
\begin{proposition}\label{PP.3}
Let $\Fq$ be a finite field with $q=p^s$ elements. Then,
\begin{enumerate}
    \item The trinomial $x^p-x-1$ is irreducible over $\Fq$ if and only if $ \gcd(s,p)=1;$ otherwise, it splits in $\mathbb{F}_q$. 
    \item For any $b\in \Fq^*,$  the trinomial $b^px^p-bx-1$ is irreducible over $\Fq$ if and only if $ \gcd(s,p)=1;$ otherwise, it splits in $\mathbb{F}_q$.
    \item In particular, if $p=2, $  the polynomial $x^2-x-1$ is irreducible over $\mathbb{F}_{2^{2k+1} }$  for any  positive integer $k;$ otherwise, it splits in $\mathbb{F}_{2^{2k} }.$
%\item If $p\neq 2$ then  $ x^p-x-1$ is irreducible over $\mathbb{F}_{p^{2k} }$ for any  positive integer $k;$ otherwise, it splits in $\mathbb{F}_{p^{2k} }$. 
\end{enumerate}
    
\end{proposition}
\begin{proof}
\begin{enumerate}
    \item  As  $\Tr_{\Fq/\mathbb{F}_p}(1)= 1+1+\ldots + 1= s$ we have that  $\Tr_{\Fq/\mathbb{F}_p}(1)=0$ in $\mathbb{F}_p$ if and only if $ s$  is a multiple of $p$. As $p$ is a prime number, we get that $\Tr_{\Fq/\mathbb{F}_p}(1)= 0$ is equivalent to $ \gcd(s,p)\neq 1$. The statement now follows from Lemma \ref{LLL.6}.
    \item Follows from the fact that if $f(x)$ is irreducible, then  $ f(bx)$ is also irreducible, for each $b\in \Fq^*$, see \cite[p. 121]{Lidl1987}.
%    \footnote{Suppose that $g(x):=f(bx)=Q(x) P(x)$ with $P$ and $Q$ are non constant polynomials. Then $f(x)=g(b^{-1} x)= Q(b^{-1} x) P(b^{-1}x),$ which is contradict to the fact that $f(x)$ is irreducible. Then we deduce that if $f(x)$ is irreducible the $ f(bx)$ is irreducible too.}
    \item[3.] is direct applications of  1. 
    \qed
    \end{enumerate}
\end{proof}

\begin{corollary}\label{Cor.3}
For each integer $\ell,$ the polynomial $ x^{p^{\ell+1}} - x^{p^{\ell}} - 1 $ can be factorized into irreducible factors over $\mathbb{F}_q$ (where $ q=p^s$) as follows:
  $$
  x^{p^{\ell+1}} - x^{p^{\ell}} - 1 =
  \begin{cases}
  \left( x^p - x - 1 \right)^{p^{\ell}}, & \text{if $\gcd(s,p)=1$ }\\
  \dprod{i=0}{p-1} (x-(\beta+i))^{p^{\ell}}, & \text{else }\\
  \end{cases}
  $$
where $ \beta\in \Fq$ is a root of $ x^p - x - 1.$
\end{corollary}

\begin{proof}
As the characteristic of $\Fq$ is $p,$ we have by Proposition \ref{PP.3} that
$$ x^{p^{\ell+1}}-x^{p^{\ell}}-1= \left(x^{p}-x-1 \right)^{p^{\ell}}=
\begin{cases}
 \left(x^{p}-x-1 \right)^{p^{\ell}} , & \text{if $\gcd(s,p)=1$ }\\
\dprod{i=0}{p-1} (x-(\beta+i))^{p^{\ell}}, & \text{else }\\
  \end{cases}
$$
 where $\beta $ is a root of $ x^p - x - 1.$
\qed 
\end{proof}

We can now derive results about the possible generator polynomials of codes associated to some specific trinomials. We start with codes of length $n=p^{\ell+1}$:

\begin{theorem}\label{Th.4}
Let $\ell$ be an integer and  $a, b$ be two elements of $\mathbb{F}_q^*$ such that $a^{p} = b^{p-1}$. Then each  $p^{\ell}$-trinomial code  associated with the polynomial $x^{p^{\ell+1}} - a x^{p^{\ell}} - b$ has a polynomial generator of the form 
$$
g(x) =   \begin{cases}
 \left((\alpha x)^{p}- \alpha x-1 \right)^{j} , & \text{if $\gcd(s,p)=1$ }\\
\dprod{i=0}{p-1} ( \alpha x-(\beta+i))^{j}, & \text{else }\\
  \end{cases} 
  \quad \text{with } 1 \leq j \leq p^{\ell},
$$
where 
%$\alpha$ is the unique scalar such that
$\alpha = (b a^{-1})^{p^{-\ell}}$, and $\beta$ a root of $ x^p-x-1.$

\end{theorem}
\begin{proof}
 As $a^{p} = b^{p-1}$, we have from Lemma \ref{L_Binomial} that
 $$\gcd(x^{p^{\ell+1} } - b, x^{p^{\ell+1} - p^{\ell}} - a) = x^{\gcd(p^{\ell+1} - p^{\ell}, p^{\ell})} - c = x^{p^{\ell}} - c = (x - \alpha)^{p^{\ell}},$$ for some $c$ and $\alpha$ in $\mathbb{F}_q^*$ such that $c = \alpha^{p^{\ell}}$. This polynomial has a root $\alpha$ with multiplicity $p^{\ell}$. Since $\alpha$ is a common root of $x^{p^{\ell+1}} - b$ and $x^{p^{\ell+1} - p^{\ell}} - a$, it follows that $\alpha^{p^{\ell+1}} = b$ and $\alpha^{p^{\ell+1} - p^{\ell}} = a$. Thus,
$$
b a^{-1} = \frac{\alpha^{p^{\ell+1}}}{\alpha^{p^{\ell+1} - p^{\ell}}} = \alpha^{p^{\ell}}.
$$
Now we can factor $x^{p^{\ell+1}} - a x^{p^{\ell}} - b$ as follows:
$$
\begin{array}{rl}
x^{p^{\ell+1}} - a x^{p^{\ell}} - b &= b \left( b^{-1} x^{p^{\ell+1}} - ab^{-1} x^{p^{\ell}} - 1 \right) \\
&= b \left( \alpha^{-p^{\ell+1}} x^{p^{\ell+1}} - \alpha^{-p^{\ell}} x^{p^{\ell}} - 1 \right) \\
&= b \left( (\alpha^{-1} x)^{p^{\ell+1}} - (\alpha^{-1} x)^{p^{\ell}} - 1 \right) \\
&= b \left( (\alpha^{-1} x)^p - \alpha^{-1} x - 1 \right)^{p^{\ell}}.
\end{array}
$$
According to Lemma \ref{LLL.6}, $x^p - x - 1$ is irreducible because $\gcd(s, p) = 1; $ otherwise, it splits in $\mathbb{F}_q$. Therefore, $x^{p^{\ell+1}} - a x^{p^{\ell}} - b$ can be factored over $\mathbb{F}_q$ as follows:
$$
x^{p^{\ell+1}} - a x^{p^{\ell}} - b = 
\begin{cases}
b \left( (\alpha^{-1} x)^p - (\alpha^{-1} x) - 1 \right)^{p^{\ell}},  & \text{if $\gcd(s,p)=1$ }\\
\dprod{i=0}{p-1} ( \alpha x-(\beta+i))^{p^{\ell}}, & \text{else }
\end{cases}.
$$
Now any $p^{\ell}$-trinomial code associated with the polynomial $x^{p^{\ell+1}} - a x^{p^{\ell}} - b$ has a generator polynomial of the  desired form.
\qed 
\end{proof}

We then turn to codes of length $n=p^{\ell+s}$, where $s$ is the extension degree of $\Fq$. For this we first need the following lemma.

\begin{lemma}\cite[Theorem. 3.80.]{Lidl1987}\label{L.5}
 For $x^q-x-a$ with $a$  an element of the subfield $\mathbb{F}_{p^r}$ of $\mathbb{F}_q, \ q=p^s$, we have the decomposition
$$
x^q-x-a= \prod_{j=1}^{q/p^r}\left(x^{p^r}-x-\beta_j\right)=\prod_{j=1}^{p^{s-r}}\left(x^{p^r}-x-\beta_j\right)
$$
in $\mathbb{F}_q[x]$, where the $\beta_j$ are the distinct elements of $\mathbb{F}_q$ with $\Tr_{\Fq / \mathbb{F}_{p^r}}\left(\beta_j\right)=a$.
\end{lemma}

\begin{theorem}\label{CCor.4}
Let $\mathbb{F}_q$ be a finite field with $q = p^s$ elements, and  $\ell$  a positive integer. Then the polynomial $x^{p^{\ell+s}} - x^{p^{\ell}} - 1$ has the following irreducible decomposition:
$$
x^{p^{\ell+s}} - x^{p^{\ell}} - 1 = \prod_{j=1}^{p^{s-1}} \left( x^p - x - \beta_j \right)^{p^{\ell}},
$$
where $\operatorname{Tr}_{\mathbb{F}_q / \mathbb{F}_p}(\beta_j) = 1$.
\end{theorem}
\begin{proof}
Since the characteristic of $\mathbb{F}_q$ is $p$, we have
$$
x^{p^{\ell+s}} - x^{p^{\ell}} - 1 = \left( x^{p^{s}} - x - 1 \right)^{p^{\ell}} = \left( x^q - x - 1 \right)^{p^{\ell}}.
$$
As $1 \in \mathbb{F}_p$, by Lemma \ref{L.5} we obtain
$$
x^{p^{\ell+s}} - x^{p^{\ell}} - 1 = \prod_{j=1}^{p^{s-1}} \left( x^p - x - \beta_j \right)^{p^{\ell}},
$$
where $\operatorname{Tr}_{\mathbb{F}_q / \mathbb{F}_p}(\beta_j) = 1$. According to Lemma \ref{LLL.6}, since $\operatorname{Tr}_{\mathbb{F}_q / \mathbb{F}_p}(\beta_j) = 1 \neq 0$, the polynomial $x^p - x - \beta_j$ is irreducible over $\mathbb{F}_q$ for each $j = 1, \ldots, p^{s-1}$.
\end{proof}

\begin{corollary}\label{Th.5}
Let $\ell$ be an integer and  $a, b$ be two elements of $\mathbb{F}_q^*$, where $\ q=p^s$. If $(a,b) \sim_{(p^{\ell+s},p^{\ell})}(1,1)$
then    $ a=1$ and   each  $p^{\ell}$-trinomial  code $C$  associated with the polynomial $x^{p^{\ell+s}}- x^{p^{\ell}}-b,$ with $b\in \Fq^{*},$ has a polynomial generator of the form 
$$ g(x)=  \prod_{j=1}^{p^{s-1}} \left( (\alpha^{-1} x)^{p}- \alpha^{-1} x-\beta_j \right)^{i}, \ \text{with} \ 0 \leq i\leq p^{\ell},  $$
for some $\beta_j\in \Fq^*$ with $\Tr_{\Fq/ \mathbb{F}_p}(\beta_j)= 1$ and  $\alpha\in \Fq^{*} $  such that   $ \alpha= b^{-p^{-\ell}}  .$
\footnote{Such an $\alpha$ always exists since the equation $ x^{p^{\ell}}=b^{-1}$ has a unique solution in $\Fq$. 
%Indeed $ \sigma(a)=a^{p^{\ell}}$ is an automorphism of $\Fq,$ and so there exist  a unique $ \alpha$ such that $ b^{-1}= \alpha^{p^{\ell}},$ which means that $ \alpha= (b^{-1})^{p^{-\ell}}= b^{-p^{-\ell}}$. 
}
\end{corollary}

\begin{proof}
 Suppose that $(a,b) \sim_{(p^{\ell+s},p^{\ell})}(1,1)$ then  by Corollary \ref{Cor.ver2} there exists an $\alpha\in \Fq^*$ such that  
 $$ (b \alpha^{p^{s+\ell}}, a \alpha^{p^{s+\ell}-p^{\ell}})= (b \alpha^{ q p^{\ell}}, a \alpha^{(p^s-1)p^{\ell}}) =(b\alpha^{p^{\ell}}, a\alpha^{(q-1)^{p^{\ell}}})=(b\alpha^{p^{\ell}},a )=(1,1).$$
 Thus $a=1$ and $ \alpha= b^{-p^{-\ell}}$. The rest of the proof is a direct application of Theorem \ref{CCor.4}.
 
\qed

\end{proof}

\section{Examples of Families of Trinomial Codes}\label{sec:examples}

\noindent In this section, we give some examples of restricting the search space for good trinomial codes by applying the theory developed before. 

\begin{example}[$3^{\ell}$-trinomial codes of length $n=3^{\ell+1}$ over $\mathbb{F}_3$ ]
We consider $3^{\ell}$-trinomial  codes of length $n=3^{\ell+1}$  over $\mathbb{F}_3. $ By Theorem \ref{Th.1}, the number of $ (3^{\ell+1},3^{\ell})$-equivalence classes is equal to
$$ N= \frac{4}{ \lcm( \frac{2}{\gcd( 3^{\ell+1},2)}, \frac{2}{\gcd( 3^{\ell+1}-3^{\ell},2)} )}= \frac{4}{\lcm( 2, \frac{2}{\gcd(2.3^{\ell} ,2)} )} = 2 .$$
As $ \gcd( \frac{2}{\gcd( 3^{\ell+1},2)}, \frac{2}{\gcd( 3^{\ell+1}-3^{\ell},2)} )= ( 2, 1)=1$ -- according to Theorem \ref{Equiv_2} -- each $3^{\ell}$-trinomial  codes of length $n=3^{\ell+1}$ is equivalent to a  $3^{\ell}$-trinomial  code associated with $ x^{3^{\ell+1}}-x^{3^{\ell}}-1 $ or $ x^{3^{\ell+1}}- 2 x^{3^{\ell}}-1. $

%Let us now  determine the classes of $(1,1)$ and $(1,2)$. 
According to Corollary \ref{Cor.ver2},
$(a_0,a_{\ell})\sim_{(3^{\ell+1}, 3^{\ell})}(1,1) $ if there is $\alpha \in \mathbb{F}_3^{*}$ such that 
 $$ (a_0,a_{\ell})\star(\alpha^{3^{\ell+1}}, \alpha^{3^{\ell+1}-3^{\ell}})=(a_0 \alpha^{3^{\ell+1}}, a_{\ell} \alpha^{2\cdot 3^{\ell}})=(a_0\alpha, a_{\ell})=(1,1).$$
% $$ (a_0,a_{\ell})\star(\alpha^{3^{\ell+1}}, \alpha^{3^{\ell+1}-3^{\ell}})=(a_0 \alpha^{3^{\ell+1}}, a_{\ell} \alpha^{2\cdot 3^{\ell}})=(a_0\alpha^{3^{\ell}}, a_{\ell})\stackrel{!}{=}(1,1).$$ 
Thus $ a_{\ell}=1$ and $ a_0= \alpha^{-1}$ for $\alpha=1,2$. Therefore, the equivalence class of $(1,1)$  consists of the   pairs 
$   ( 1, 1)$  and $(2 , 1)$.
Similarly the  class of  $(1,2)$ consists of the pairs
$  ( 1,2)$ and $ (2 , 2)  .$

Since $ \gcd(3,1)=1,$ by Corollary \ref{Cor.3}  the factorization of $x^{3^{\ell+1}} - x^{3^{\ell}} - 1$ is given by
$$ x^{3^{\ell+1}} - x^{3^{\ell}} - 1= (x^3-x-1)^{3^{\ell}} .$$  
According to Theorem \ref{Th.4}, each code in the class  of $(1,1)$ has a generator polynomial of the form $g(x)=  (x^3-x-1)^{i}$, for some $ 0 \leq  i \leq 3^{\ell}$.  
This class contains some optimal codes, for example, by taking   $\ell=2,$ we found that the $9$-trinomial code associated with $  x^{27} + 2x^9 + 2$ and generated by 
$g(x)=x^{24} + x^{22} + x^{21} + x^{20} + 2x^{19 }+ 2x^{18} + x^{16} + 2x^{15} + x^{14} + x^{11 }+ x^9
    + x^8 + x^7 + x^6 + 2x^5 + 2x^4 + x^2 + 2x + 1$
    is an optimal  $[27,3,18]_3$-code.\footnote{These codes attain the Griesmer bound.}

 For $\ell=3,$ we found that the $27$-trinomial code associated with $  x^{27} + 2x^9 + 2$ and generated by 
$g(x)=x^{78} + x^{76} + x^{75} + x^{74} + 2x^{73} + 2x^{72} + x^{70} + 2x^{69} + x^{68} + x^{65} + x^{63} + x^{62} + x^{61} + 2x^{60} + 2x^{59} + x^{57} + 2x^{56} + x^{55} + x^{52} + x^{50} + x^{49} + x^{48} + 2x^{47} + 2x^{46} + x^{44} + 2x^{43} + x^{42} + x^{39} + x^{37} + x^{36} + x^{35} + 2x^{34} + 2x^{33} + x^{31} + 2x^{30} + x^{29} + x^{26} + x^{23} + x^{21} + x^{20} + x^{19} + 2x^{18} + 2x^{17} + x^{15} + 2x^{14} + x^{13} + x^{10} + x^8 + x^7 + x^6 + 2x^5 + 2x^4 + x^2 + 2x + 1$
is an optimal $[81,3,55]_3$-code.\footnotemark[3]

 \end{example}

 \begin{example}[$3^{\ell}$-trinomial codes of length $n=3^{\ell+2}$ over $\mathbb{F}_9$]
 We consider $3^{\ell}$-trinomial codes of length $n=3^{\ell+2}$ over $ \mathbb{F}_{9}= \mathbb{F}_{3}(\xi)$ with $\xi $ a primitive element of $\mathbb{F}_9$. The number of $(3^{\ell+2},3^{\ell})$-equivalence classes is 
 $$ N= \frac{64}{ \lcm( \frac{8}{\gcd( 3^{\ell+2},8)}, \frac{8}{\gcd( 3^{\ell+2}-3^{\ell},8)} )}  = \frac{64}{\lcm( 8, \frac{8}{\gcd(8. 3^{\ell},8)} )}= \frac{64}{\lcm( 8, 1 )} = 8,$$ 
 and each class contains $8$ pairs. Since $\gcd( \frac{8}{\gcd(3^{\ell+2},8)}, \frac{8}{\gcd(8. 3^{\ell},8)} )= \gcd( 8,1 )=1,$ then by Theorem \ref{Equiv_2}, the $8$ possible pairs are 
 $$ (1,1), (\xi,1), (\xi^2,1), (\xi^3,1), (\xi^4,1),(\xi^5,1), (\xi^6,1), (\xi^7,1). $$
It follows that for $a,b \in \mathbb{F}_9^*$ the  $3^{\ell}$-trinomial code family associated with $x^{3^{\ell+2}}-  a x^{3^{\ell}}- b,$ is equivalent to a $3^{\ell}$-trinomial code family associated with one of the polynomials   $x^{3^{\ell+2}}-  \xi^j x^{3^{\ell}}-1, \ j=0,1,\ldots, 7.$

We now determine the class of $(1,1)$. According to Corollary \ref{Cor.ver2}, $(a_0,a_{\ell})\sim_{(3^{\ell+2}, 3^{\ell})}(1,1)$ if there exists $\alpha \in \mathbb{F}_9^{*}$ such that  
$$
(a_0,a_{\ell})\star(\alpha^{3^{\ell+2}}, \alpha^{3^{\ell+2}-3^{\ell}}) = (a_0 \alpha^{3^{\ell+2}}, a_{\ell} \alpha^{8 \cdot 3^{\ell}}) = (a_0\alpha^{3^{\ell}}, a_{\ell})=
 (1,1).
$$
%$$
%(a_0,a_{\ell})\star(\alpha^{3^{\ell+2}}, \alpha^{3^{\ell+2}-3^{\ell}}) = (a_0 \alpha^{3^{\ell+2}}, a_{\ell} \alpha^{8 \cdot 3^{\ell}}) = (a_0\alpha^{3^{\ell}}, a_{\ell}) = (1,1).
%$$
Thus, $a_{\ell}=1$ and $a_0= \alpha^{-3^{\ell}}$. Therefore, the class of $(1,1)$ is  given by 
$$
 \{ (\alpha^{-3^{\ell}} , 1) : \alpha\in \mathbb{F}_9^{*} \} =  \{ ( 1, 1), (\xi,1), (\xi^2,1), (\xi^3,1), (\xi^4,1), (\xi^5,1), (\xi^6,1), (\xi^7,1) \}.
$$
Note that the second equality follows from the fact that the order of $\mathbb{F}_9^*$ is $8$ and thus coprime to $3^\ell$ for any $\ell$, which implies that all elements of $\mathbb{F}_9^*$ appear in the first coordinate.

 According to  Theorem \ref{CCor.4},  the factorization of $x^{3^{\ell+2}} - x^{3^{\ell}} - 1$ is given by:
$$ x^{3^{\ell+2}} - x^{3^{\ell}} - 1= (x^9-x-1)^{3^{\ell}} = \prod_{j=1}^{3} \left( x^3 - x - \beta_j \right)^{3^{\ell}},
$$
for some $\beta \in \Fq^*$ with $\operatorname{Tr}_{\mathbb{F}_9 / \mathbb{F}_3}(\beta_j) = 1$. 
It follows that 
$$ x^{3^{\ell+2}} - x^{3^{\ell}} - 1= \left( x^3 - x - \xi \right)^{3^{\ell}} \left( x^3 - x - 2 \right)^{3^{\ell}} \left( x^3 - x - \xi^3 \right)^{3^{\ell}} $$
Hence, the $3^{\ell}$-trinomial codes of length $n=3^{\ell+2}$ over $ \mathbb{F}_{9}$ which are  equivalent to  $3^{\ell}$-trinomial codes associated with $ x^{3^{\ell+2}}-x^{3^{\ell}}-1$
will have a generator polynomial of the form  
$$g(x)=  \left( x^3 - x - \xi \right)^{i} \left( x^3 - x - 2 \right)^{j} \left( x^3 - x - \xi^3 \right)^{h}  ,\quad 0 \leq  i,j,h \leq 3^{\ell}.$$ 
%\textcolor{purple}{Can we find some good code examples here, or are the parameters too big?}
%\hlc[green]{ I didn't find much good codes here, in comparison to the existing ones in Grassel's table; I found these just this  two ones }

For $\ell=1, $ we constructed the $ 3-$trinomial code associated to $x^{27}-x^3-1$,  and generated by $g(x)=x^3 -x + \xi^7$ which is an optimal $[27,24,3]_9$-code.\footnote{These codes attain an upper bound on the minimum distance according to \cite{databases}.}

For $\ell=2, $ we constructed the $ 9-$trinomial code associated to $x^{81}-x^9-1$,  and generated by $g(x)=x^3 -x + \xi^5$ which is an optimal $[81,78,3]_9$-code.\footnotemark[4]

\end{example}

\begin{example}[Trinomial codes of length $27$ over $\mathbb{F}_4$]
%\textcolor{purple}{Is there a reason why we take the length $3^3$ over $\mathbb F_4$?}
%\hlc[green]{There is no a special reason. But we took in consideration that $ n=27$ is a multiple of $q-1=3$ in the computation of the number of equivalence classes. }
We consider the  case of $\ell$-trinomial codes of length $27$ over $\mathbb{F}_4= \mathbb{F}_2(\xi),$ with $\xi $ a primitive element of $ \mathbb{F}_4$.   The number of $(27, \ell)$-equivalence classes is given by
$$
\frac{3^2}{\lcm\left(\frac{3}{\gcd(27, 3)}, \frac{3}{\gcd(27-\ell, 3)}\right)} = \frac{9}{\lcm\left(1, \frac{3}{\gcd(27-\ell, 3)}\right)}.
$$
We have two cases:

\begin{enumerate}
    \item If $\ell \equiv 0 \pmod{3}$: 
    The number of $(27, \ell)$-equivalence classes is $9$. In this case, we consider all pairs $(a, b)$ from $\mathbb{F}_4^* \times \mathbb{F}_4^*$.
    
    \item If $\ell \not\equiv 0 \pmod{3}$: 
    The number of $(27, \ell )$-equivalence classes is $3$. Since
    $$
    \gcd\left(\frac{3}{\gcd(27, 3)}, \frac{3}{\gcd(27-\ell, 3)}\right) = \gcd(1, 3) = 1,
    $$
   then -- by Theorem \ref{Equiv_2} -- the representatives of these three $(n, \ell)$-equivalence classes are the pairs $(1, 1)$, $(1, \xi)$, and $(1, \xi^2)$. So we need to consider the three polynomials
   $ x^{27}-x^{\ell}-1,\ \ x^{27}-x^{\ell}-\xi , $ and $x^{27}-x^{\ell}-\xi^2.$
   
   For $\ell=8 ,$ we found that the $8$-trinomial code associated to $ x^{27} - x^8 - \xi^2$ and generated by 
   $ x^6 + \xi^2 x^5 + \xi^2x^3 + \xi x^2 + x +1 $ is a $[27, 21, 4]_4$-code, which equals the best known parameters according to \cite{databases}.
   
   For $\ell=5,$ the $5$-trinomial code associated with $ x^{27} - x^5 - 1$ and generated by $ x^{10} + x^9 + x^8 + x^7 + x^6 + \xi^2x^5 + x^4 +\xi^2 x^2 + \xi x + 1$ is a $ [27,17,6]_4 $-code, which equals the best known parameters according to \cite{databases}.
\end{enumerate}

\end{example}

\section{Equivalence of Polycyclic Codes}\label{sec:eq_poly}
 In this section we generalize the results on equivalence to general polycyclic codes. We start with the generalized definition of equivalence, for which we will use the notation $\vec{a}(x) := a_0 + a_1 x + \ldots + a_{n-1} x^{n-1}$.
 
 \begin{definition}
Let $ \vec{a} = (a_0, a_1, \ldots, a_{n-1}) $ and $ \vec{b} = (b_0, b_1, \ldots, b_{n-1}) $ be elements in $\mathbb{F}_q^n .$
We say that $ \vec{a} $ and $ \vec{b} $ are \textit{$ n $-equivalent}, and we denote this by 
$$ \vec{a} \sim_n \vec{b}, $$ 
%\textcolor{purple}{Shall we use $\overrightarrow{n}$ or something like this, so that the notation is different from the constacyclic case?}
%\hlc[green]{ For constacyclic codes, it was denoted by $ \lambda \sim_n \mu,$ but here I used symbol vector "$ \vec{a} \sim_n \vec{b} $" I think it good to conserve tgis notation.  }
if there exists an 
$ \alpha \in \mathbb{F}_q^* $ such that the following map

  \begin{equation}
  	\begin{array}{cccc}
  	\varphi_{\alpha}:& \mathbb{F}_q[x] /\langle x^n-\vec{b}(x) \rangle  &\longrightarrow & \mathbb{F}_q[x] /\langle x^n-\vec{a}(x) \rangle, \\ 
  	& f(x) & \longmapsto &  f(\alpha x),
  	\end{array}
  	\end{equation} 
is an $\Fq$-algebra isomorphism. 
 \end{definition}
 Note that, as before, the map $\varphi_\alpha$ is a Hamming isometry.  Moreover, we can easily verify that $ " \sim_n "$ is an equivalence relation.
 %on $(\Fq^{*})^m,$ where $ m$ is the weight of $\vec{a}.$

 \begin{remark}
 \begin{enumerate}
     \item  If $ \vec{a} = (\lambda, 0, \ldots, 0) $ and $ \vec{b} = (\mu, 0, \ldots, 0) $, then we recover the case of $n$-equivalence for constacyclic codes studied in \cite{Chen2014}.
     \item  If $ \vec{a} = (a_0, 0, \ldots,a_{\ell},0,\ldots, 0) $ and  $ \vec{b} = (b_0, 0, \ldots,b_{\ell},0,\ldots, 0) $, then we recover the case of $(n,\ell)$-equivalence for $\ell$-trinomial  codes studied in Section 3.
     \end{enumerate}
 \end{remark}

 In the following we show that this notion of equivalence automatically implies that the vectors $\vec{a}$ and $\vec{b}$ have zero entries in exactly the same position. This implies that $\ell$-trinomial code families can only be equivalent to other $\ell$-trinomial code families.
 
 \begin{lemma}
    Let $ \vec{a} = (a_0, a_1, \ldots, a_{n-1}) $ and $ \vec{b} = (b_0, b_1, \ldots, b_{n-1}) $ be elements in $\mathbb{F}_q^n $ such that  $ \vec{a} \sim_n\vec{b} $. Then $ a_i \neq 0 $ if and only if $ b_i \neq 0 $, for any $ 0\leq  i\leq n-1,$  and so $\vec{a}$ and $\vec{b}$ have the same Hamming weight.
 \end{lemma}
\begin{proof}
Suppose that $ \vec{a} \sim_n \vec{b},$ then there is $\alpha \in \Fq^*$ such that $ \varphi_{\alpha}$ is an $\Fq$-algebra isometry between $ \mathbb{F}_q[x] /\langle x^n-\vec{b}(x) \rangle $  and $ \mathbb{F}_q[x] /\langle x^n-\vec{a}(x) \rangle$. Then, as in the proof of [Theorem \ref{Th.1}, (1) $\Rightarrow$ (2)], we obtain that 
$$ b_i = \alpha^{n-i} a_i, \quad \forall i=0, \ldots, n-1.$$
Hence the result holds.
    \qed
\end{proof}

We now generalize Theorem \ref{Th.1} to the general polycyclic case: 

\begin{theorem}\label{Th_general}
Let $ \vec{a}=(a_0,a_1,\ldots, a_{n-1}), \vec{b}=(b_0,b_1,\ldots, b_{n-1})\in \mathbb{F}_q^{n}$ have non-zero entries in the same $m$ positions, i.e., $a_{i_j}$ and $b_{i_j}$ are non-zero  for $0\leq i_0<\dots<i_{m-1}\leq n-1$. Moreover, let $\xi $ be a primitive element of $\Fq$. Then the following statements are equivalent:
\begin{enumerate}
\item $ \vec{a} \sim_{n} \vec{b}$. 

\item The polynomials $ a_{i_j} x^{n-i_j} - b_{i_j} \in \mathbb{F}_q[x]$, for $ j \in \{0,1,\ldots,m-1\}$, have a common root in $\mathbb{F}_q^*$.

\item The polynomial $ \gcd( a_{i_0} x^{n-i_0} - b_0, a_{i_1} x^{n-i_1} - b_{i_1}, \ldots, a_{i_{m-1}} x^{n-i_{m-1}} - b_{i_{m-1}})$ has at least one root in $\mathbb{F}_q^*$.

\item The polynomial $ \gcd_{\{0 \leq j \leq m-1\}}( x^{n-{i_j}} - b_{i_j} a_{i_j}^{-1}) $ has at least one root in $\mathbb{F}_q^*$.
 \item  There exists $\alpha\in \Fq^*$ such that $$ (a_{i_0}, a_{i_1},\ldots, a_{i_{m-1}})\star(\alpha^{n-{i_0}}, \alpha^{n-{i_1}} , \ldots, \alpha^{n-{i_{m-1}}})=  ( b_{i_0}, b_{i_1},\ldots, b_{i_{m-1}}).$$
\item $ (a_{i_0}, a_{i_1},\ldots, a_{i_{m-1}})^{-1}\star( b_{i_0}, b_{i_1},\ldots, b_{i_{m-1}}) \in H,$ where $H$ is the cyclic subgroup of $ (\Fq^*)^m$ generated by 
$( \xi^{n-{i_0}},\xi^{n-{i_1}},\ldots, \xi^{n-{i_{m-1}}} ).$
\end{enumerate}
In particular the number of $n$-equivalence classes is
$$ N=\dfrac{ (q-1)^m}{ \lcm_{ (0 \leq j\leq m-1)}\left(\frac{q-1}{ \gcd(n-i_j,q-1)}\right)}. $$ 
\end{theorem}

\begin{proof}
\begin{enumerate}
    \item [(1) $\Rightarrow$ (2)] 
    Suppose that $ \vec{a} \sim_{n} \vec{b} $, then there is $\alpha \in \mathbb{F}_q^*$ such that

$$\varphi_{\alpha} : \mathbb{F}_q[x] /\langle x^n - \vec{b}(x) \rangle \to \mathbb{F}_q[x] /\langle x^n - \vec{a}(x) \rangle, \quad f(x) \mapsto f(\alpha x) $$
is an $\mathbb{F}_q$-algebra isometry.  It follows that

$$ 
\varphi_{\alpha}(x^k) =  \alpha^k x^{k} \mod (x^n - \vec{a}(x)), \quad \forall  \ k = 0, 1, \ldots, n-1.
$$

As $\varphi_{\alpha}$ is an $\mathbb{F}_q$-algebra isometry and $ \varphi(x^n - \vec{b}(x)) = 0 \mod (x^n - \vec{a}(x))$, then
\begin{equation}
\varphi_{\alpha}(x^n) = \varphi_{\alpha}(\vec{b}(x)) = b_0 + \alpha b_1 x + \ldots + \alpha^{n-1}b_{n-1} x^{n-1}, \mod (x^n - \vec{a}(x)).
\end{equation}

On the other hand,
\begin{equation}
\varphi_{\alpha}(x^n) = \alpha^n x^n = \alpha^n ( a_0 + a_1 x + \ldots + a_{n-1} x^{n-1}), \mod (x^n - \vec{a}(x))
\end{equation}

Comparing term by term, we deduce that for any $ i \in \{0, 1, \ldots, n-1\}, \ a_i \alpha^{n-i} = b_i$, which means that $\alpha$ is a common root of the polynomials $ a_i x^{n-i} - b_i $, for $ i \in \{0, 1, \ldots, n-1\}$. As  $ a_{i_j}$'s and $  b_{i_j}$'s are the non-zeros components of $\vec{a}$ and $\vec{b},$   then $ a_{i_j} \alpha^{n-{i_j}} = b_{i_j},\ j=0,\ldots,m-1$.

\item[(2) $\Rightarrow$ (3)] and  (3) $\Rightarrow$ (4) are immediate. 

 \item [(4) $\Rightarrow$ (5)]
    Let $\alpha $ be a root of the polynomial   $\gcd_{\{0 \leq i \leq n-1, \ a_i \neq 0\}}( x^{n-i} - b_i a_i^{-1}) $. Then $\alpha$ is a common root of the polynomials $  x^{n-i_j} - b_{i_j} a_{i_j}^{-1} \  \text{for any} \ j\in \{ 0,1,\ldots, m-1\},$ and so   $a_{i_j} \alpha^{n-i_j}= b_{i_j}$.
    It follows that 

$$  (b_{i_0}, b_{i_1},\ldots, b_{i_{m-1}} )=  (\alpha^{n-i_0}, 
\alpha^{n-i_1} , \ldots, \alpha^{n-i_{m-1}})\star( a_{i_0},  a_{i_1}, \ldots, a_{i_{m-1}}).$$ 
\item [(5) $\Rightarrow$ (6)]
Suppose that there is $\alpha\in \Fq^{*}$ such that 
$$  (b_{i_0}, b_{i_1},\ldots, b_{i_{m-1}} )=  (\alpha^{n-i_0}, 
\alpha^{n-i_1} , \ldots, \alpha^{n-i_{m-1}})\star( a_{i_0},  a_{i_1}, \ldots, a_{i_{m-1}}).$$ 
For $\alpha= \xi^h,$ we obtain that 
$$  ( a_{i_0}, a_{i_1},\ldots, a_{i_{m-1}})^{-1}\star(b_{i_0}, b_{i_1},\ldots, b_{i_{m-1}} ) =   (\xi^{n-{i_0}}, \xi^{n-{i_1}} , \ldots, \xi^{n-i_{m-1}})^h .$$
It follows that $  ( a_{i_0},  a_{i_1}, \ldots, a_{i_{m-1}})^{-1}\star( b_{i_0},  b_{i_1}, \ldots, b_{i_{m-1}}) $ belongs to the cyclic subgroup  $H$ of $ (\Fq^{*})^m$ generated by $ (\xi^{n-{i_0}}, \xi^{n-{i_1}} , \ldots, \xi^{n-{i_{m-1}}}) .$
    \item [(6) $\Rightarrow$ (1)]  
     Suppose that $ (a^{n-{i_0}}, a^{n-{i_1}} , \ldots, a^{n-i_{m-1}})^{-1}\star(b^{n-{i_0}}, b^{n-{i_1}} , \ldots, b^{n-i_{m-1}}) $ is an element of the cyclic subgroup  $H$ of $ (\Fq^{*})^m$ generated by
$ (\xi^{n-{i_0}}, \xi^{n-{i_1}} , \ldots, \xi^{n-{i_{m-1}}}) $. Then there exists an integer $ h$ such that 
     $$  ( a_{i_0},  a_{i_1}, \ldots, a_{i_{m-1}})^{-1}\star( b_{i_0},  b_{i_1}, \ldots, b_{i_{m-1}}) =  (\xi^{n-{i_0}}, \xi^{n-{i_1}} , \ldots, \xi^{n-{i_{m-1}}})^h $$
    For $ \beta= \xi^{h}, $ we obtain  that $  a_j \beta^{n-j}= b_j, $ for any $ j\in \{ i_0,i_1,\ldots, i_{m-1}\}.$
As in the proof of Theorem \ref{Th.1}, we verify that  $ \varphi_{\beta}, $ as follows: 
  \begin{equation}
        \begin{array}{cccc}
        \varphi_{\beta} : & \mathbb{F}_q[x]/\langle x^n - \vec{b}(x) \rangle,  & \longrightarrow & \mathbb{F}_q[x] /\langle x^n - \vec{a}(x) \rangle, \\ 
        & f(x) & \longmapsto & f(\beta x),
        \end{array}
    \end{equation} 
	
	is an $\Fq$-algebra isometry with respect to the Hamming distance. 
\end{enumerate}
\noindent By the equivalence between (1) and (6), we deduce that the number of $n$-equivalence classes on on $(\Fq^*)^m$ corresponds to the order of the group $ (\Fq^*)^m/H,$ which equals $$ N=\dfrac{ (q-1)^m}{ \lcm_{ (0 \leq j\leq m-1)}\left(\frac{q-1}{ \gcd(n-i_j,q-1)}\right)}. $$ 

\qed
\end{proof}

Similarly to the case of $\ell$-trinomial codes, Theorem \ref{Th_general} implies the following results regarding the equivalence of polycyclic codes. The proofs are analogous to the trinomial case. 

\begin{corollary}\label{Equiv_General2}
As before let $ \vec{a}=(a_0,a_1,\ldots, a_{n-1})$ and $ \vec{b}=(b_0,b_1,\ldots, b_{n-1})$ be elements of $\mathbb{F}_q^{n}$  
of the same weight $m$  and denote by $ a_{i_j}$ and $  b_{i_j}$  the non-zeros components of $\vec{a}$ and $\vec{b}$. 
\begin{enumerate}
    \item      The class of polycyclic  codes associated with the polynomial $ x^n-\dsum{j=0}{m-1} a_{i_j} x^{i_j} $ is equivalent to the class of  polycyclic  codes associated with the polynomial  $x^n-\dsum{j=0}{m-1}  x^{i_j}$  if and only if there exists $\alpha \in \Fq^*$ such that 
    $$ ( a_{i_0}, a_{i_1},\ldots, a_{i_{m-1}})\star( \alpha^{n-i_0}, \alpha^{n-i_1},\ldots, \alpha^{n-i_{m-1}}) = (1, 1,\ldots, 1 ).$$  
    \item  Let  $ d=\gcd_{ 0\leq j\leq m-1} \left( \frac{q-1}{ \gcd(n-i_j,q-1)}\right) ,$ then the  class of polycyclic  codes associated with the  polynomial $ x^n-\dsum{j=0}{m-1} a_{i_j} x^{i_j} $ is equivalent to the class of  polycyclic  codes associated with the polynomial  $ x^n-\dsum{j=0}{m-1} \xi^{k_j} x^{i_j} ,$ for $k_j=0,1,\ldots, \gcd(n-i_j,q-1)-1 ,$  where $\xi$ is a primitive element of $ \Fq$. 
  %  \item The number of  $\Fq$-algebra isometries $\varphi_{\alpha}$ between $ \Fq[x]/\langle x^n-\vec{b}(x)\rangle $ and $ \displaystyle{ \Fq[x]/\langle x^n-\vec{a}(x)\rangle}$ is $$ \gcd(n-i_0, n-i_1,\ldots, n-i_{m-1}, q-1) .$$
\end{enumerate}
    \end{corollary} 
    \begin{proof}
        \begin{proof}
        
\begin{enumerate}

\item Follows from the fifth assertion of  Theorem of \ref{Th_general}.
\item Let $d= \gcd_{ 0\leq j\leq m-1} \left( \frac{q-1}{ \gcd(n-i_j,q-1)}\right).$
\begin{itemize}
    \item  If $ d = 1 $, then the cyclic group $ H $ generated by $ (\xi^{n-{i_0}}, \xi^{n-{i_1}} , \ldots, \xi^{n-{i_{m-1}}}) $ is isomorphic to the group $ \langle \xi^{n-{i_0}} \rangle \times \langle \xi^{n-{i_1}} \rangle \times \ldots \times \langle \xi^{n-{i_{m-1}}} \rangle $ and has order $ \frac{(q-1)^m}{  d_0 d_{1} \ldots d_{m-1}}, $ with $ d_j=\gcd(n-i_j,q-1)$ for $j=0,\ldots, m-1$. By Theorem \ref{Th_general}, the number of $ n$-equivalence classes is $ d_0 d_{1} \ldots d_{m-1} $. Therefore, we can partition $ (\Fq^*)^m $ as 
$$
(\Fq^*)^m =  \bigcup_{k_0=0}^{d_0-1} \bigcup_{k_{1}=0}^{d_{1}-1} \ldots \bigcup_{k_{m-1}=0}^{d_{m-1}-1} (\xi^{k_0}, \xi^{k_1},\ldots, \xi^{k_{m-1}}) H.
$$
So the result holds.
 \item 
 If $d\neq 1, $ the number of  $(n,\ell)$-equivalence classes is $ d d_0 d_{1} \ldots d_{m-1},$ and so we partition $ (\Fq^*)^m  $ as 
 
$$
(\Fq^*)^m  = \bigcup_{h=0}^{d-1} \bigcup_{k_0=0}^{d_0-1} \bigcup_{k_{1}=0}^{d_{1}-1} \ldots \bigcup_{k_{m-1}=0}^{d_{m-1}-1} (\xi^{k_0+hn}, \xi^{k_1},\ldots, \xi^{k_{m-1}}) H.
$$
which implies the result.
\end{itemize}

\end{enumerate}
  \qed
\end{proof}
    \end{proof}

\begin{example}[Polycyclic codes of length $n=12$ over $\mathbb{F}_3$ ] 
 We consider polycyclic  codes of length $n=12$ over $ \mathbb{F}_{3}$ associated with  a polynomial of the form $f(x)=x^{12}- cx^7-bx-a\in \mathbb{F}_3[x]$.  Denote $\vec{a}(x):=  cx^7+bx+a $, then  the Hamming weight of $ \vec{a}(x) $ is  $3$ and according to Theorem  \ref{Th_general}, the number of $12$-equivalence classes is 
 $$ N= \frac{2^3}{ \lcm( \frac{2}{\gcd( 12,2)}, \frac{2}{\gcd( 12-1,2)}, \frac{2}{\gcd( 12-7,2)} ) }  =
 \frac{2^3}{ \lcm( 1, 2, 2 ) } = 4.$$ 
 Since $$\gcd(  \frac{2}{\gcd( 12,2)}, \frac{2}{\gcd( 12-1,2)}, \frac{2}{\gcd( 12-7,2)}  )= \gcd( 1,2,2 )=1,$$ then by Corollary \ref{Equiv_General2}, 
 each  polycyclic codes associated with a polynomial $f$ of the form $f =x^{12}- cx^7-bx-a,$ is equivalent to a polycyclic  code  associated with one of the following  polynomials: 
 $$f_1= x^{12}-x^7-x-1, \quad f_2= x^{12}- x^{7}-x-\xi, \quad f_3= x^{12}-  \xi x^{7}-x-1, \quad  f_4= x^{12}-  \xi x^{7}-x-\xi,$$
 where $\xi=2$ is (the only) primitive element of $\mathbb{F}_3$. We then searched for good codes in the corresponding spaces. According to Codes Tables \cite{databases} some of these codes are optimal  (for given $n,k$ and $q$); we present these code parameters in  Table \ref{Tab1}.
 %,  we collected best known polycyclic codes, according to Codes Tables in \cite{databases},  that we can construct in the corresponding space.

 \begin{table}[h]
		\small{	\begin{tabular}{|l |l|c|}
						\hline 
		Class polynomial $f(x)$		&	 Generator polynomial of the polycyclic code & Parameters\\
						\hline 
	
  $ f_1(x)=x^{12}-x^7-x-1$&	$ x^{10} + 2x^8 + x^6 + 2x^5 + 2x^4 + x^3 + x^2 +2x + 2 $ &  $ [12, 2, 9]_{3} $\\
							\cline{2-3}
		  	 	&	$ x^8 + x^6 + 2 x^3 + 2x^2 + 2x + 2 $   &  $ [12, 4, 6]_{3}$\\
					\hline 					
											
    $ f_3(x)=x^{12} - 2  x^{7} -x-1 $	 	&	$  x^2 + 2x + 2  $   &  $ [12, 10, 2]_{3}$\\
							\cline{2-3}
						&	$ x^7 + x^6 + 2x^5 + x^4 + 2x^2 + 2 $& $ [12, 5, 6]_{3} $\\
           \cline{2-3}
						&	$ x^5 + x^3 + x^2 + 2 x + 1 $& $ [12, 7, 4]_{3} $\\
     
    \cline{2-3}
						&	$  x^4 + x^3 + 2 x^2 + 2 $& $ [12, 8, 3]_{3} $\\
   \hline 					
											
    $ f_4(x)=x^{12} - 2  x^{7} -x-2  $	 	&	$  x^3 + x^2 + 2 $   &  $ [12, 9, 3]_{3}$\\                     
                    
	\hline 						
\end{tabular} }

\caption{Optimal polycyclic codes of length $n=12$ over   $\mathbb{F}_{3} $.}
\label{Tab1}
\end{table}
\end{example}

\begin{example}[Polycyclic codes of length $n=15$ over $\mathbb{F}_4$ ] 

 We consider skew polycyclic  codes of length $n=15$ over $ \mathbb{F}_{4}=\mathbb{F}_2(\xi),\ \xi$ a primitive element of $\mathbb{F}_4,$  associated with  a polynomial of the form $f(x) =x^{15}-cx^h-bx^l-a\in \mathbb{F}_4[x],$ with $l,h$  are integers such that $0<l<h <15$.  Let $\vec{a}(x):=  cx^h+ bx^l+a ,$  the Hamming weight of $ \vec{a}(x) $ is  $3,$ then according to Theorem  \ref{Th_general}, the number of $(15,\sigma)$-equivalence classes is 
 $$ N= \frac{3^3}{ \lcm( \frac{3}{\gcd( [15]_1 ,3)}, \frac{3}{\gcd( 2^l[15-l]_1,3)}, \frac{3}{\gcd( 2^h[15-h]_1,3)} ) }  =
 \frac{3^3}{ \lcm( 1, \frac{3}{\gcd( 15-l,3)}, \frac{3}{\gcd( 15-h,3)} ) } .$$
 \begin{enumerate}
     \item If $l$ and $h$ are multiples of $3$ then the number if equivalence classes is $N=3^3$ and the equivalence relation has no influence in this case.
     \item Else, i.e., $l$  or   $h$ is not a multiple of $3$ then $N=\dfrac{3^3}{ 3}= 9$. Let suppose that $h$ is a multiple and $l$ not, then as  
     $$  \gcd\left( \frac{3}{\gcd( 15 ,3)}, \frac{3}{\gcd( 15-l,3)}, \frac{3}{\gcd( 15-h,3)} \right)= \gcd( 1,3, 1)= 1,$$ then by Corollary \ref{Equiv_General2}, 
 each  polycyclic codes associated with a polynomial $f$ of the form  $f =x^{15}-cx^h-bx^l-a,$ is equivalent to a polycyclic  code  associated with one of the following  polynomials: 
 $$f_{i,j}= x^{15}- \xi^i x^h-  x^l -\xi^j, \quad i,j \in \{0,1, 2\}.  $$
We again found some optimal (for given $n,k$ and $q$) polycyclic codes, which we present in  Table \ref{Tab2}.
%c, for different values of $i,j, h$ and $l$ some of optimal linear codes, according to Codes Tables in \cite{databases}, that we construct from these classes of polycyclic codes.
 \begin{table}[h]
    \centering
  \small{  \begin{tabular}{|l|l|l|}
        \hline
       \textbf{Polynomials $f$} &  \textbf{Generator Polynomial} &  \textbf{Parameters} \\
        \hline
     $x^{15} + x^3 + x^2 + \xi$ &    $x^3 + x^2 + \xi^2$ &  $[15, 12, 3]_4$ \\
        \hline
       $x^{15} + \xi x^3 + x^2 + \xi$ &  $x^5 + \xi^2 x^3 + x + \xi^2$ &  $[15, 10, 4]_4$ \\
        \hline
        $x^{15} + \xi x^3 + x^2 + \xi$ &  $x^4 + x^3 + \xi x^2 + \xi x + \xi^2$ & $[15, 11, 4]_4$ \\
        \hline
        $x^{15} + x^6 + x^2 + \xi$ & $x^{10} + x^8 + \xi^2 x^7 + \xi x^6 + \xi x^5 + \xi^2 x^4 + \xi^2 x^3 + \xi x^2 + \xi x + 1$ &  $[15, 5, 8]_4$ \\
        \hline
      $x^{15} + \xi x^{12} + x^2 + 1$ &  $x^6 + \xi x^5 + x^4 + \xi x^3 + \xi x^2 + \xi x + \xi^2$ &  $[15, 9, 5]_4$ \\
        \hline
        $x^{15} + \xi^2 x^{12} + x + \xi^2$ & $x^4 + x^3 + \xi^2 x^2 + \xi^2 x + \xi$ &  $[15, 11, 4]_4$ \\
        \hline
    \end{tabular}}
 \caption{Optimal polycyclic codes of length $n=15$ over   $\mathbb{F}_{4} $.}
    \label{Tab2}
\end{table}

 \end{enumerate}

\end{example}

\section{Conclusion}\label{sec:conclusion}

In this work, we developed a framework for studying equivalence among polycyclic codes via an extension of the notion of $n$-equivalence. Our results provide explicit criteria for determining when two ambient spaces of polycyclic codes are equivalent and yield a classification of these spaces into finitely many equivalence classes. A detailed analysis was carried out for $\ell$-trinomial codes, where we derived explicit formulas for the number of $n$-equivalence classes and established concrete equivalence conditions, including a reduction to a canonical trinomial under mild arithmetic assumptions. These results were then extended to broader families of polycyclic codes, demonstrating that the proposed approach applies beyond special cases. In addition, we constructed examples of trinomial polycyclic codes with optimal parameters.

From a practical perspective, the classification obtained in this paper allows one to restrict to representatives of $n$-equivalence classes, both in future algebraic analyses and in the search for codes with good parameters. Beyond classification and parameter optimization, the study of equivalence between polycyclic code ambient spaces provides a suitable framework for understanding duality of polycyclic codes, since the dual of a polycyclic code is not necessarily polycyclic with respect to the same ambient space, but is often equivalent to a polycyclic code arising from a different ambient space \cite{Ouazzou2023,Shi2023,Aydin2022}. A systematic understanding of the equivalences of the ambient spaces may therefore lead to new results on the duality of polycyclic codes.

Several directions for further research arise naturally. One is the explicit factorization of trinomials defining polycyclic codes via their order, in analogy with the well-understood cases of cyclic and constacyclic codes. This may open the possibility of constructing infinite families of trinomial or polycyclic codes with optimal or near-optimal parameters. Another promising direction is the investigation of more general classes of isometries, such as maps of the form $\varphi_{\alpha}(x)=\alpha x^{k}$, which may lead to a finer equivalence theory and a deeper structural understanding of polycyclic codes.

In this work, we developed a framework for studying equivalence among polycyclic codes via an extension of the notion of $n$-equivalence. Our results provide explicit criteria for determining when two ambient spaces of polycyclic codes are equivalent and yield a classification of these spaces into finitely many equivalence classes.
A detailed analysis was carried out for $\ell$-trinomial codes, where we derived explicit formulas for the number of $n$-equivalence classes and established concrete equivalence conditions, including a reduction to a canonical trinomial under mild arithmetic assumptions. These results were then extended to broader families of polycyclic codes, demonstrating that the proposed approach applies beyond special cases. In addition, we constructed examples of trinomial polycyclic codes with optimal parameters.\\
From a practical perspective, the classification obtained in this paper allows one to restrict  to representatives of $n$-equivalence classes, both in future algebraic analyzes and in the search for codes with good parameters.\\
Several directions for further research arise naturally. One is the explicit factorization of trinomials defining polycyclic codes via their order, in analogy with the well-understood cases of cyclic and constacyclic codes. This may open the possibility of constructing infinite families of trinomial or polycyclic codes with optimal or near-optimal parameters. Another promising direction is the investigation of more general classes of isometries, such as maps of the form $\varphi_{\alpha}(x)=\alpha x^{k}$, which may lead to a finer equivalence theory and a deeper structural understanding of polycyclic codes.

\section*{Acknowledgments}
The authors thank Adrien Pasquereau, Abdelghefar Chibloun, and Prof.\ Nuh Aydin for their many helpful discussions and suggestions to improve this paper.

\section*{Declarations}
The first author’s research is supported by a Swiss Government Excellence Scholarship (ESKAS), no:  2024.0504.
 \bibliographystyle{abbrv}
\bibliography{Bibleo.bib}

\end{document}